\documentclass[10pt,twocolumn,twoside]{IEEEtran}
\usepackage[utf8]{inputenc}
\usepackage{amsmath}
\renewcommand{\theequation}{S.\arabic{equation}}
\renewcommand{\thesection}{S.\Roman{section}}

\usepackage{mathrsfs}

\usepackage{amssymb}
\usepackage{braket}
\usepackage{amsmath,mathtools}
\usepackage{float}
\usepackage{graphicx}
\usepackage{enumerate}
\usepackage{bbm}
\DeclareMathOperator{\Tr}{Tr}
\usepackage{verbatim}
\usepackage{accents}
\usepackage{authblk}
\usepackage{xr}
\usepackage{hyperref}
\usepackage{amsmath, amssymb, bbm, xspace}
\usepackage{epsfig}
\usepackage{longtable}
\usepackage{color}
\usepackage{mathrsfs}
\usepackage{comment}
\usepackage{ifthen}
\newboolean{showcomments}
\setboolean{showcomments}{true}
\usepackage{courier}
\usepackage{xcolor}
\usepackage{todonotes}
\usepackage[framemethod=tikz]{mdframed}
\usepackage{lineno}
\newmdenv[leftmargin=\dimexpr-0.4em, innerleftmargin=0.5em,
rightmargin=\dimexpr-0.4em, innerrightmargin=0.5em,
linewidth=2pt,linecolor=red, topline=false, bottomline=false,
innertopmargin=0pt,innerbottommargin=0pt,skipbelow=0pt,skipabove=0pt,%
]{notex}

\newenvironment{note}%
{\vskip\dimexpr\dp\strutbox-\prevdepth\relax\notex\strut\ignorespaces}%
{\xdef\notetpd{\the\prevdepth}\endnotex\vskip-\notetpd\relax}

\let\oldtodo\todo

\makeatletter%
\DeclareDocumentCommand{\todo}{ O{} +g +d<> }{%
		\setlength{\marginparwidth}{1.5cm}%
	\IfNoValueTF{#2}{\relax}{%
		\oldtodo[caption={#2},size=\scriptsize,#1]{\renewcommand{\baselinestretch}{0.8}\selectfont\sffamily#2\par}%
	}%
	\IfNoValueTF{#3}{\relax}{%
		\IfNoValueTF{#2}{
			\begin{note}%
				\begin{internallinenumbers}%
					\indent%
					#3%
				\end{internallinenumbers}%
			\end{note}%
		}{
			\vspace{-0\baselineskip}%
			\begin{note}%
				\begin{internallinenumbers}%
					\indent%
					#3%
				\end{internallinenumbers}%
			\end{note}%
		}%
	}%
}%
\makeatother
\usepackage{soul}

\newcommand{\hlc}[2][yellow]{{%
		\colorlet{foo}{#1}%
		\sethlcolor{foo}\hl{#2}}%
}

\newcommand{\removetodo}[2]{\todo[color=pink]{\textbf{delete:} ``#1'' #2}\hlc[pink]{#1}}
\newcommand{\inserttodo}[1]{\todo[color=green!40]{\textbf{insert:} #1}}

\newcommand{\hltodoy}[2]{\todo[color=yellow!40]{#2}\hl{#1} }

\newcommand{\hltodoc}[3]{\todo[color=#3!40]{#2}\hlc[#3]{#1} }

\newcommand{\hltodo}[2]{\todo[color=orange!40]{#2}\hlc[orange!40]{#1} }
\newcommand{\replacetodo}[2]{\todo[color=pink!40]{\textbf{replace with:}``#2'' }\hl{#1} }

\usepackage{marginnote}
\newcommand{\todol}[1]{{%
		\let\marginpar\marginnote
		\reversemarginpar
		\renewcommand{\baselinestretch}{0.8}%
		\todo{#1}}}
	
\newcommand{\inserttodol}[1]{{%
		\let\marginpar\marginnote
		\reversemarginpar
		\renewcommand{\baselinestretch}{0.8}%
		\inserttodo{#1}}}
	
\newcommand{\removetodol}[2]{{%
		\let\marginpar\marginnote
		\reversemarginpar
		\renewcommand{\baselinestretch}{0.8}%
		\removetodo{#1}{#2}}}

\newcommand{\hltodol}[2]{{%
		\let\marginpar\marginnote
		\reversemarginpar
		\renewcommand{\baselinestretch}{0.8}%
		\hltodo{#1}{#2}}}

\newcommand{\replacetodol}[2]{{%
		\let\marginpar\marginnote
		\reversemarginpar
		\renewcommand{\baselinestretch}{0.8}%
		\replacetodo{#1}{#2}}}

\newcommand{\hltodoyl}[2]{{%
		\let\marginpar\marginnote
		\reversemarginpar
		\renewcommand{\baselinestretch}{0.8}%
		\hltodoy{#1}{#2}}}

\newcommand{\hltodocl}[3]{{		\let\marginpar\marginnote
		\reversemarginpar
		\renewcommand{\baselinestretch}{0.8}%
		\hltodoc{#1}{#2}{#3}}}

\newtheorem{theorem}{Theorem}[section]

\newtheorem{lemma}[theorem]{Lemma}

\newtheorem{proposition}[theorem]{Proposition}

\newtheorem{corollary}[theorem]{Corollary}

\newtheorem{definition}{Definition}[section]

\newcommand{\comb}[2]{\prescript{#1}{}C_{#2}}

\def\bkE{{\rm I\kern-.17em E}}
\def\bk1{{\rm 1\kern-.17em l}}
\def\bkD{{\rm I\kern-.17em D}}
\def\bkR{{\rm I\kern-.17em R}}
\def\bkP{{\rm I\kern-.17em P}}

\def\bkZ{{\bf{Z}}}

\def\bkE{{\rm I\kern-.17em E}}
\def\bk1{{\rm 1\kern-.17em l}}
\def\bkD{{\rm I\kern-.17em D}}
\def\bkR{{\rm I\kern-.17em R}}
\def\bkP{{\rm I\kern-.17em P}}

\makeatletter
\newcommand{\pushright}[1]{\ifmeasuring@#1\else\omit\hfill$\displaystyle#1$\fi\ignorespaces}
\newcommand{\pushleft}[1]{\ifmeasuring@#1\else\omit$\displaystyle#1$\hfill\fi\ignorespaces}
\makeatother

\def\bkZ{{\bf{Z}}}
\def\b12{(\beta_1,\beta_2)}

\newenvironment{proofarg}[1][]{\noindent\hspace{2em}{\itshape Proof #1: }}{\hspace*{\fill}~\qed\par\endtrivlist\unskip}
\newenvironment{example}{{\noindent \bf Example}}{\hfill $\square$\hspace{-4.5pt}\vspace{6pt}}
\newcounter{example}
\renewcommand{\theexample}{\thesection.\arabic{example}}

\newcounter{remark}
\renewcommand{\theremark}{\thesection.\arabic{remark}}

\def\t{^\top}
\def\Bscr{\mathscr{B}}
\def\Xscr{\mathcal{X}}

\newlength{\noteWidth}
\long\def\notes#1{\ifinner
{\tiny #1}
\else
\marginpar{\parbox[t]{\noteWidth}{\raggedright\tiny #1}}
\fi\typeout{#1}}

 \def\notes#1{\typeout{read notes: #1}} 

\newcommand{\ie}{i.e.\@\xspace} 

\newcommand{\inv}{^{-1}}

\def\Pbb{{\mathbb{P}}}

\def\Ibf{{\bf I}}

\def\ast{\buildrel{t}\over\rightarrow}

\def\dim{\mathop{\hbox{\rm dim}}}

\def\Span{\mbox{\rm span}}

\def\half  {{\textstyle{1\over 2}}}

\def\conv{\textrm{conv}\:}

\def\spose#1{\hbox to 0pt{#1\hss}}

\def\text #1{\hbox{\quad#1\quad}}

\def\nthinsp{\mskip -2   mu}

\def\superstar{^{\raise 0.5pt\hbox{$\nthinsp *$}}}
\def\SUPERSTAR{^{\raise 0.5pt\hbox{$*$}}}

\def\lamstarT {\lambda^{\raise 0.5pt\hbox{$\nthinsp *$}T}}

\def\Bscr{{\cal B}}

\def\Hscr{{\cal H}}
\def\Lscr{{\cal L}}

\def\Pscr{{\cal P}}
\def\Qscr{{\cal Q}}
\def\Sscr{{\cal S}}

\def\Uscr{{\cal U}}

\def\Nscr{{\cal N}}

\def\Cscr{{\cal C}}

\def\Xscr{{\cal X}}

\def\aur{\;\textrm{and}\;}

\def\bfI{{\bf I}}
\let\forallnew\forall
\renewcommand{\forall}{\forallnew\ }
\let\forall\forallnew

		\def\bkE{{\rm I\kern-.17em E}}
		\def\bk1{{\rm 1\kern-.17em l}}
		\def\bkD{{\rm I\kern-.17em D}}
		\def\bkR{{\rm I\kern-.17em R}}
		\def\bkP{{\rm I\kern-.17em P}}
		\def\bkY{{\bf \kern-.17em Y}}
		\def\bkZ{{\bf \kern-.17em Z}}
		\def\bkC{{\bf  \kern-.17em C}}

{\begin{list}{}%
         {\setlength{\leftmargin}{#1}}%
         \item[]%
}
{\end{list}}

\def\Tsf{\mathsf{T}}
\def\Esf{\mathsf{E}}

\def\Rsf{\mathsf{R}}
		\def\bsp{\begin{split}}
		\def\beq{\begin{eqnarray}}
		\def\bal{\begin{align*}}
		\def\bc{\begin{center}}
		\def\be{\begin{enumerate}}
		\def\bi{\begin{itemize}}
		\def\bs{\begin{small}}
		\def\bS{\begin{slide}}
		\def\ec{\end{center}}
		\def\ee{\end{enumerate}}
		\def\ei{\end{itemize}}
		\def\es{\end{small}}
		\def\eS{\end{slide}}
		\def\eeq{\end{eqnarray}}
		\def\eal{\end{align*}}
		\def\esp{\end{split}}
		\def\qed{ \vrule height7.5pt width7.5pt depth0pt}  

	\def\cp2problem#1#2#3#4{\fbox
		 {\begin{tabular*}{0.9\textwidth}
			{@{}l@{\extracolsep{\fill}}l@{\extracolsep{6pt}}l@{\extracolsep{\fill}}c@{}}
				#1 & & $#4 $ 
			\end{tabular*}}}

		\def\bkE{{\rm I\kern-.17em E}}
		\def\bk1{{\rm 1\kern-.17em l}}
		\def\bkD{{\rm I\kern-.17em D}}
		\def\bkR{{\rm I\kern-.17em R}}
		\def\bkP{{\rm I\kern-.17em P}}
		
		\def\bkZ{{\bf{Z}}}

\newcommand {\beeq}[1]{\begin{equation}\label{#1}}
\newcommand {\eeeq}{\end{equation}}
\newcommand {\bea}{\begin{eqnarray}}
\newcommand {\eea}{\end{eqnarray}}

\def\texitem#1{\par\smallskip\noindent\hangindent 25pt
               \hbox to 25pt {\hss #1 ~}\ignorespaces}

\def\bsp{\begin{split}}
		\def\beq{\begin{eqnarray}}
		\def\bal{\begin{align*}}
		\def\bc{\begin{center}}
		\def\be{\begin{enumerate}}
		\def\bi{\begin{itemize}}
		\def\bs{\begin{small}}
		\def\bS{\begin{slide}}
		\def\ec{\end{center}}
		\def\ee{\end{enumerate}}
		\def\ei{\end{itemize}}
		\def\es{\end{small}}
		\def\eS{\end{slide}}
		\def\eeq{\end{eqnarray}}
		\def\eal{\end{align*}}
		\def\esp{\end{split}}
		\def\qed{ \vrule height7.5pt width7.5pt depth0pt}  

\usepackage{amsmath, amssymb, xspace}
\usepackage{epsfig}
\usepackage{longtable}
\usepackage{color}
\usepackage{mathrsfs}
\usepackage{subfig}
\def\Cscr{{\cal C}}
\def\Nscr{{\cal N}}

\def\conv{{\rm conv}}





\makeatletter
\def\thetable{S.\Roman{table}}
\makeatother

\reversemarginpar
\title{\bf The Quantum Advantage in Binary Teams and the Coordination Dilemma:  Supplementary Material}
\author[2]{Shashank A. Deshpande}
 \author[1]{Ankur A. Kulkarni}
\affil[2]{Department of Aeronautics and Astronautics, Massachusetts Institute of Technology, Cambridge, MA, 02139, USA.}
\affil[1]{Systems and Control Engineering,  Indian Institute of Technology Bombay, Mumbai, 400076, India. \\ Emails: \texttt{croshank@mit.edu, kulkarni.ankur@iitb.ac.in}}
\date{}

\begin{document}
\maketitle

\begin{abstract}
    This document contains supporting material for our paper ``The Quantum Advantage in Binary Teams and the Coordination Dilemma'' \cite{deshpande2022binmatcombined}.
\end{abstract}

\section{Introduction}
This document contains supporting material for the paper ``The Quantum Advantage in Binary Teams and the Coordination Dilemma'' \cite{deshpande2022binmatcombined}. The supplementary material is structured into sections which develop the proofs of results from the main paper. In some cases, the proofs themselves require further supporting results, which are also derived herein. All equations and environments in the supplementary material are numbered as `S.xxx', to distinguish them from the ones from the main paper. For example equation (11)  refers to  equation number 11 the from main paper, and (S.11) refers to the one in the supplementary material.

\textit{A note on reading the supplementary material:} The supplementary material is not a standalone document and must be seen as an appendix to the main paper. The reader is advised to follow the main document for notations and other developments, which are not all recalled here, and refer to the supplementary material only for skipped proofs. These proofs can be read off directly from the relevant section of the supplementary material in continuity with the main document.
\section{Proof of Results in Section \ref{sec:equivalences}}\label{app:equiproofs}
\subsection{Proof of Proposition \ref{prop:transpose}}
We prove (1) $\implies$ (2) for arbitrary $M,N$; replacing $M,N$ with $M\t,N\t$ the reverse implication will follow.
Suppose $\Cscr(M,N)$ admits a quantum advantage and consider a problem instance $D:=(M, N, \Pbb, \Uscr_A,\Uscr_B, \chi)$ with cost function $ \ell $
and a quantum strategy $Q$ $:=(\rho_{AB}$,$\{ P_{u_A}^{(A)}(\xi_A)\}_{u_A,\xi_A}$, $\{ {P}^{(B)}_{u_B}(\xi_B)\}_{u_B, \xi_B})$.

Now consider an instance  $D'$ in the class $\Cscr(M\t, N\t)$ given by  $(M\t, N\t, \Pbb', \Uscr_A',\Uscr_B', \chi')$ where $\chi'  = \chi,
\Uscr_A'  = \Uscr_B,
\Uscr_B'  = \Uscr_A $ and
$\Pbb'(\xi_A, \xi_B, \xi_W)  \equiv \Pbb(\xi_B, \xi_A, \xi_W). $
It follows that the cost $\ell'$ of $D'$ satisfies $\ell'(u_A',u_B',\xi_W) \equiv \ell(u_B',u_A',\xi_W).$
Consider a strategy $$Q':=(\rho_{AB}\t, \{{P}'^{(A)}_{u_A'}(\xi_A)\}_{u_A', \xi_A}, \{{P}'^{(B)}_{u_B'}(\xi_B)\}_{u_B', \xi_B} )$$
where ${P}'^{(A)}_{u_A'}(\xi_A)=P^{(B)}_{u_A'}(\xi_A),\quad  \forall u_A' \in \Uscr_A', \xi_A \in \Xi_A,
{P}'^{(B)}_{u_B'}(\xi_B)=P^{(A)}_{u_B'}(\xi_B),\quad \forall u_B' \in \Uscr_B', \xi_B \in \Xi_B$. Now by properties of the trace, for all $\xi_A \in \Xi_A$ ,$\xi_B \in \Xi_B$, $u_A' \in \Uscr_A'$, $u_B'\in \Uscr_B'$,
$$\Tr\rho_{AB}\t {P}'^{(A)}_{u_A'}(\xi_A){P}'^{(B)}_{u_B'}(\xi_B) = \Tr\rho_{AB} {P}^{(A)}_{u_B'}(\xi_B) {P}^{(B)}_{u_A'}(\xi_A).$$
Thus $Q'$ satisfies for all $u_A' \in \Uscr_A', u_B'\in \Uscr_B'$,
$Q'(u_A',u_B'|\xi_A,\xi_B) = Q(u_B',u_A'|\xi_B,\xi_A)$
and \eqref{eq:qcost} (from the main paper)

\begin{align*}
J(Q';D')
&=	\sum_{\xi_A, \xi_B, \xi_W}\Pbb'(\xi_A, \xi_B, \xi_W)\sum_{ u_A', u_B'} \ell'(u_A', u_B', \xi_W)\\
&\qquad \qquad \times\Tr\left(\rho_{AB}P_{u_B'}^{(A)}(\xi_B)P_{u_A'}^{(B)}(\xi_A)\right) \\
&= \sum_{\xi_A, \xi_B, \xi_W}\Pbb(\xi_B, \xi_A, \xi_W)\sum_{ u_A', u_B'} \ell(u_B', u_A', \xi_W)\\
& \qquad\times  \Tr\left(\rho_{AB}P_{u_B'}^{(A)}(\xi_B)P_{u_A'}^{(B)}(\xi_A)\right) = J(Q;D),
\end{align*}
where we have used that $u_A'\in \Uscr_A'=\Uscr_B$ and $u_B'\in \Uscr_B'=\Uscr_A.$
Similarly for any deterministic strategy $\gamma\in\Gamma$ consider a $\gamma'\in\Gamma'$ for the instance $D'$ such that $\gamma'_A(\xi_A)=\gamma_B(\xi_A)$ and $\gamma'_B(\xi_B)=\gamma_A(\xi_A)$  so that
$J(\pi_{\gamma'};D')=J(\pi_{\gamma};D),$
and hence $J^*_\Lscr(D)=J^*_\Lscr(D').$
Thus if $\exists Q\in \Qscr $ such that $  J(Q; D)<J^*_\Lscr(D)$ then with $Q',D'$ as above, we have $J(Q'; D')< J^*_\Lscr( D'),$ whereby $\Cscr(M\t,N\t)$ admits a quantum advantage.
We have established the proposition.

\subsection{Proof of Proposition \ref{prop:permute}}
($(1)\Leftrightarrow (3)$) We will show that $(3)$ implies $(1)$. But since $\Rsf^2=\Ibf$, applying the same to result with $M,N$ replaced by $\Rsf M,\Rsf N$, we can conclude that $(1) \implies (3)$.
 For $D= (M, N, \Pbb, \Uscr_A, \Uscr_B, \chi) \in \Cscr(M, N)$, consider a $D'=(\mathsf{R}M, \mathsf{R}N, \Pbb, \Uscr_A', \Uscr_B, \chi) \in \Cscr(\mathsf{R}M, \mathsf{R}N)$ and let its cost function be denoted $\ell'$. We think of $\Uscr_A $ as a column vector $(u_A^0,u_A^1)$ and let $\Uscr_A' = \Rsf \Uscr_A$.
Then for any strategy $Q \in \Pscr(\Uscr|\Xi)$, we have, $J(Q; D')$
\begin{align*}
 &=\sum_{\xi, u_B}\Pbb(\xi)\sum_{u_A'\in \Uscr_A'}\ell'(u_A^{\prime}, u_B, \xi_W)Q(u_A^{\prime}, u_B|\xi_A, \xi_B)\\
&=\sum_{\xi, u_B}\Pbb(\xi)\sum_{u_A \in \Uscr_A}\ell(u_A, u_B, \xi_W)Q(u_A, u_B|\xi_A, \xi_B)
\end{align*}
$=J(Q; D)$ where $\xi = (\xi_A,\xi_B,\xi_W).$ It is easy to also see that $J_\Pi^*(D)=J_\Pi^*(D').$ Thus,
$J(Q, D)-J_\Pi^*(D)<0\implies J(Q, D')-J_\Pi^*(D)<0$
and we have shown that $(3) \implies (1)$, and thus $(1) \Leftrightarrow (3).$ The equivalence between $(2)$ and $(3)$ can be shown in a similar manner, thereby completing the proof.

\subsection{Proof of Proposition \ref{prop:exchange}}
It suffices to show $(1) \implies (2)$ for arbitrary $M,N$, following which the reverse claim will follow.
Let $D:=(M, N, \Pbb, \Uscr_A, \Uscr_B, \chi)$ and $Q$ $:=(\rho_{AB}$,$\{ P_{u_A}^{(A)}(\xi_A)\}_{u_A,\xi_A}$,$\{ {P}^{(B)}_{u_B}(\xi_B)\}_{u_B, \xi_B})\in \Qscr$.
Consider a problem $D':=(N, M, \Pbb', \Uscr_A, \Uscr_B,1/\chi) \in \Cscr(N,M)$ where $\Pbb'(\xi_A, \xi_B, 0)=\Pbb(\xi_A, \xi_B, 1), \quad \Pbb'(\xi_A, \xi_B, 1)=\Pbb(\xi_A, \xi_B, 0)$.\\
Notice that the cost function $\ell'(.)$ of the instance $D'$ is related to $\ell(.)$ of $D$ as  $\ell'(\cdot, \cdot, 0)\equiv (1/\chi)\ell(\cdot, \cdot, 1),\quad
\ell'(\cdot,\cdot, 1)\equiv (1/\chi)\ell(\cdot,\cdot, 0)$
Thus, denoting $\xi=(\xi_A,\xi_B,\xi_W),$ we have for any $Q\in \Pscr(\Uscr|\Xi)$
\begin{align}
J(Q;D')&=\sum_{u, \xi}\Pbb'(\xi)\ell'(u_A, u_B, \xi_W)Q(u_A, u_B|\xi_A, \xi_B)\nonumber\\
&=(1/\chi)\sum_{u, \xi} \Pbb(\xi)\ell(u_A, u_B, \xi_W)Q(u_A, u_B|\xi_A, \xi_B)\nonumber\\
&=(1/\chi)J(Q;D) \nonumber
\end{align}
In particular taking $Q \in \Pi$, we get that $J^*_\Lscr(D) = \frac{1}{\chi}J^*_\Lscr(D')$. Thus if $\exists Q\in \Qscr$ such that  $J(Q;D)-J^*_\Lscr(D)<0$, then for this $Q$, we have
$J(Q; D')-J^*_\Lscr(D')=\frac{1}{\chi}(J(Q;D)-J^*_\Lscr(D))<0.$
This establishes the proposition.

\section{Supporting Results for Proof of Theorem \ref{thm:main}}\label{sec:eliminate_suppl}
Theorem \ref{thm:main} claims that a problem class admits a quantum advantage if and only if it lies in the orbit of the CAC class or the $\half$-CAC class. In Section \ref{sec:eliminate} of the paper we showed that the CAC class and $\half$-CAC class admits a quantum advantage.
In the sections below we systematically eliminate all classes not in the orbit of the CAC and $\half$-CAC class to show Theorem~\ref{thm:main}.
\def\Crsfs{\mathscr{C}}
Additionally, the derivation of the optimal deterministic policy for the $\half$-CAC instance was skipped. We include it below.

\begin{definition}
	We call a problem class $\Cscr(M,N)\in\Crsfs$ an $m$-$n$ class if the number of non-zero entries in $M$ is $m\in\{0,1,2,3,4\}$ and the number of non-zero entries in $N$ is $n\in\{0,1,2,3,4\}$ and call $V=(M, N)$ an $m$-$n$ tuple. Let $\Cscr_{mn} \subset \Crsfs$ denote the set of all $m$-$n$ problem classes.
\end{definition}
Notice that $|\Cscr_{mn}|=\prescript{4}{}C_m \prescript{4}{}C_n$ and $\{\Cscr_{mn}\}_{m,n}$ defines a partition on $\Crsfs$. For $(M, N)$ in the CAC form, $\Cscr((M, N);\Omega)\subset \Cscr_{22}$. Similarly, for $(M, N)$ in the $\half$-CAC form, $\Cscr((M, N);\Omega)\subset \Cscr_{12}\cup \Cscr_{21}$. To proceed with our systematic elimination, we eliminate $\Cscr_{mn}$ for all pairs $(m, n)\notin \{(2,2), (1,2), (2,1)\}$ through a pigeonhole principle-based argument; this is done in Section \ref{sec:cent}. In the subsequent sections, we eliminate classes in $\Cscr_{22}$, $\Cscr_{12}$ and $\Cscr_{21}$ that do not belong in the orbit of the CAC or the $\half$-CAC class.

\subsection{Problem classes with no centralisation advantage} \label{sec:cent}
Since $\ell$ takes only binary values, if there exists a pair of actions $u_A^*, u_B^*$ such that $\ell(u_A^*, u_B^*,\xi_W)$ is nonzero for both values of $\xi_W$, then the strategy
\begin{equation}\label{eq:ccq}
	\bar{Q}(u_A,u_B|\xi_A,\xi_B) \equiv  \delta(u_A=u_A^*,u_B=u_B^*)
\end{equation}
 which lies in $\Pi$ is optimal over $\Pscr(\Uscr|\Xi)$. In other words, the problem admits no centralisation advantage. The following definition and the proposition that succeeds formalises this line of arguments.

\begin{definition}
We call a pair $V=(M, N)$ overlapping if $\exists i, j $ such that $ [M]_{ij}=[N]_{ij}=-1$. Denote the set of all classes $\Cscr(V)$ where $V$ is overlapping by $\Cscr^o$,
\begin{equation}\label{eq:codef}
	\Cscr^o:=\{\Cscr(M, N)|\exists i, j: [M]_{ij}=[N]_{ij}=-1\}.
\end{equation}
\end{definition}
\begin{proposition}\label{prop:centralproblems}
If $\Cscr(M,N)\in\Cscr^o$, then $\Cscr(M,N)$ does not admit a centralisation, and hence a quantum advantage.
\end{proposition}
\begin{proof}
Let $\Cscr(M,N) \in \Cscr^o$ and let $D \in \Cscr(M,N).$
By \eqref{eq:codef}, there exists $u_A^* \in \Uscr_A,u^*_B \in \Uscr_B$ such that
\begin{equation}
	\ell({u}^{*}_A, {u}^{*}_B, \xi_W)\leq \ell(u_A, u_B, \xi_W) \quad \forall u_A,u_B,\xi_W. \label{eq:generalcentral}
\end{equation}
Thus for any $Q\in\Pscr(\Uscr|\Xi)$,
\begin{align*}
&J(Q; D)\\
&\geq \sum_{\xi_{W}} \Pbb(\xi_W|\xi_A, \xi_B)\sum_{\xi_A, \xi_B}\Pbb(\xi_A, \xi_B)\min_{u_A, u_B}\ell(u_A, u_B, \xi_W)\\
&= J(\bar{Q};D),
\end{align*}
where $\bar{Q}$ is as defined in \eqref{eq:ccq}. Since $\bar{Q}\in \Pi,$
we get $J^{**}(D)=J_\Lscr^*(D)$ and the proposition is established.
\end{proof}

Although \eqref{eq:codef} gives a tractable definition of $\Cscr^o$, it is not straightforward to exhaustively enumerate subclasses in $\Cscr^o.$ Hence we will use Proposition~\ref{prop:centralproblems} as an enabling lemma to eliminate some subclasses $\Cscr_{mn} \in \Crsfs.$
Following are two results that accomplish this.
\begin{corollary}
	Let either $m=0$ or $n=0$ and let $\Cscr(M,N) \in \Cscr_{mn}$. Then $\Cscr(M, N)$ does not admit a quantum advantage.\label{cor:nonull}
\end{corollary}
\begin{proof}
	If one of the matrices $M$ and $N$ is null, then there exist $u_A^*,u_B^*$ such that
 \eqref{eq:generalcentral} holds. The rest follows as in Proposition \ref{prop:centralproblems}.
\end{proof}
\begin{corollary}
Let $\Cscr(M,N)$ be such that $m+n\geq 5$. Then $\Cscr(M, N)$ does not admit quantum advantage.\label{cor:pigeonhole}
\end{corollary}
\begin{proof}
If $m+n\geq 5$, then by the pigeonhole principle, $\Cscr(M, N)\in \Cscr^o$.
\end{proof}

\subsection{Elimination of other problem classes}
For $i\in\{1, 2\}$, let $-i$ denote the element in $\{1, 2\}\setminus\{i\}$.
\begin{definition}
	We call $V = (M,N)$ and the class $\Cscr(V)$ achiral if $V$ is non-overlapping and $\exists i, j $ such that $ [M]_{ij}=[N]_{-i-j}=-1$. We call a $V$ and the class $\Cscr(V)$ chiral if $V$ is non-overlapping and not achiral.
\end{definition}
The following lemma shows that the properties of  $V$ being overlapping and V being achiral extend to its orbit.
\begin{lemma}\label{lem:nooverlapgen}
	(1) $V$ is overlapping if and only if $V'$ is overlapping for all $V'\in(V;\Omega)$.
	(2) $V$ is achiral if and only if $V'$ is achiral for all $V'\in (V;\Omega)$.
\end{lemma}
\begin{proofarg}[of Lemma \ref{lem:nooverlapgen}]\label{app:nooverlapgen}
(1) It is easy to see that by inspection all actions in $\Omega$ map an overlapping pair $(M,N)$ to another overlapping pair. Moreover, since $M,N$ are $2\times 2$ matrices, all actions $\Rsf,\Rsf',\Tsf,\Esf$ are involutions, \ie, when applied twice, are equivalent to $\bfI$. In other words if $V' \in (V;\Omega)$, then by a suitable application actions, one can map $V'$ to back to $V$, whereby if $V'$ is overlapping, then so must be $V.$ \\
(2) This part follows in a similar manner as (1).
Suppose that $V$ is achiral. Then owing to part (a), every $V'$ in the orbit $(V;\Omega)$ is non-overlapping. We will show that the action $\Rsf$ preserves achirality of $V$; this can be shown for other actions can be proved similarly.
Let $i, j$ be such that $[M]_{ij}=[N]_{-i-j}=-1$. Then,  $[\Rsf M]_{-ij}=[M]_{ij}=[N]_{-i-j}=[\Rsf N]_{i-j}=-1$, implying that $\Rsf V$ is achiral. Thus, the orbit $(V;\Omega)$ is achiral. Again, using that the actions in $\Omega$ are involutions we get that if any $V'\in (V,\Omega)$ is achiral, then so is $V.$
\end{proofarg}

Corollaries \ref{cor:nonull} and \ref{cor:pigeonhole} help eliminate the possibility of a quantum advantage for all $m$-$n$ classes where $m+n\geq 5$ or $\min(m, n)=0$. Thus, out of the $256$ classes in $\Crsfs$, we have eliminated
$\sum_{m+n\geq 5} {^4C_m^4C_n}+\sum_{\min(m, n)=0} {^4C_m^4C_n}=93+31=124$
classes. We now scan through remaining elements in $\Crsfs$, namely,  $\Cscr_{11}, \Cscr_{12}, \Cscr_{21}, \Cscr_{22},\Cscr_{13},\Cscr_{31}$.

Our elimination procedure for $\Cscr_{mn}$ can be described as follows. We define $\Cscr_{mn}^o=\Cscr_{mn}\cap \Cscr^o$ as the collection of all overlapping $m$-$n$ classes. Observe that
\begin{equation}\label{eq:overlap}
	|\Cscr_{mn}^o| = \comb{4}{m} \times (\sum_{k=1}^{n}\comb{m}{k}\comb{4-m}{n-k} ),
\end{equation}
since we have $\comb{4}{m}$ choices for a `$-1$' in $M$ following which we have $\comb{m}{k}\comb{4-m}{n-k}$ choices for $k$ overlapping $-1$'s in $N$ and $\comb{4-m}{n-k}$ for the remaining $(n-k)$ nonoverlapping $-1$'s.
We then explicitly specify a chiral $V_c=(M_c, N_c)$ and an achiral $V_a=(M_a, N_a)$ and define $\Cscr_{mn}^c:=\Cscr((V_c;\Omega))\cap \Cscr_{mn}$, $\Cscr_{mn}^a:=\Cscr((V_a;\Omega))\cap\Cscr_{mn}$. Following Lemma \ref{lem:nooverlapgen}, such a specification ensures that all classes in $\Cscr_{mn}^c$ are chiral and all those in $\Cscr_{mn}^a$ are achiral so that $\Cscr_{mn}^o,\Cscr_{mn}^a,\Cscr_{mn}^c$ are mutually disjoint.
We then establish that our choice of $V_c$, $V_a$ ensures that $\Cscr^o_{mn}, \Cscr^a_{mn}$ and $\Cscr^c_{mn}$ exhaust $\Cscr_{mn},$ whereby these constitute a partition of $\Cscr_{mn}$.
We then examine $\Cscr(V_c)$ and $\Cscr(V_a)$ and eliminate those that do not admit a quantum advantage.

\subsubsection{Elimination of 1-1 problem class $\Cscr_{11}$}
Consider $\Cscr_{11}$ and notice $|\Cscr_{11}|=\prescript{4}{}C_1\prescript{4}{}C_1=16$. Define
$\Cscr_{11}^o=\Cscr_{11}\cap\Cscr^o.$
Define the following achiral pair $V_a=(M_a, N_a)$,
\begin{equation}\label{eq:11ma}
 M_a:=\begin{pmatrix}
-1&0\\
0&0
\end{pmatrix},\quad N_a:=\begin{pmatrix}
0&0\\
0&-1
\end{pmatrix},
\end{equation}
and let $\Cscr_{11}^a:=  \Cscr((V_a;\Omega))\cap \Cscr_{11}$. Observe that,
\begin{equation}\label{eq:c11adef}
\Cscr_{11}^a=\{\Cscr(V_a), \Cscr(\Rsf V_a),
 \Cscr(V_a\Rsf), \Cscr(\Rsf V_a\Rsf)\},
\end{equation}
so that $|\Cscr_{11}^a|=4$. Now take the chiral pair $V_c=(M_c, N_c),$
\begin{equation}\label{eq:11mc}
M_c:=\begin{pmatrix}
-1&0\\
0&0
\end{pmatrix}, N_c:=\begin{pmatrix}
0&-1\\
0&0
\end{pmatrix},
\end{equation}
and let $\Cscr_{11}^c := \Cscr((V^c;\Omega))\cap \Cscr_{11}$. It is easy to verify that
\begin{align}\label{eq:c11cdef}
\Cscr_{11}^c=&\{\Cscr(V_c), \Cscr(\Rsf V_c), \Cscr(V_c\Rsf), \Cscr(\Rsf V_c\Rsf), \Cscr(\Tsf V_c),\nonumber\\
&\Cscr(\Rsf \Tsf V_c), \Cscr(\Tsf V_c\Rsf), \Cscr(\Rsf \Tsf V_c\Rsf)\},
\end{align}
whereby $|\Cscr_{11}^c|=8 = |\Cscr_{11}|-|\Cscr_{11}^o|-|\Cscr_{11}^a|.$ Thus $\Cscr_{11}^c,\Cscr_{11}^o,\Cscr_{11}^a$ is a partition of $\Cscr_{11}.$
The following proposition eliminates $\Cscr_{11}$ by elimination of each of the elements in this partition.
\begin{proposition} \label{prop:11eliminate} $\Cscr_{11}$ does not admit a quantum advantage.
\end{proposition}
\begin{proof}
$\Cscr_{11}^o$ does not admit a quantum advantage
since $\Cscr_{11}^o\subset \Cscr^o$.
We now show the same for $\Cscr_{11}^a.$ For an instance $D\in\Cscr(M_a, N_a)$, for $(M_a,N_a)$ as defined in \eqref{eq:11ma}, note that  $\ell(u_A^0,u_B^0,0)=-1$, $\ell(u_A^1, u_B^1, 1)=-\chi$ and $\ell(.)\equiv 0$ otherwise. Now consider deterministic policies $\hat{\gamma},\overline{\gamma}:$ $ \hat{\gamma}_A(\xi_A)\equiv u_A^0, \hat{\gamma}_B(\xi_B)\equiv u_B^0 \aur \overline{\gamma}_A(\xi_A)\equiv u_A^1,\overline{\gamma}_B(\xi_B)\equiv u_B^1.$ following which it is easy to evaluate $ J(\pi_{\hat{\gamma}};D) = - \Pbb(\xi_w=0), \quad J(\pi_{\overline{\gamma}};D)=-\chi \Pbb(\xi_W=1) $.
Now consider for a no-signalling vertex $Q^{\alpha\beta\delta}\in\Nscr\Sscr$ (recall \eqref{eq:nonlocalvertex} from the main paper),
\begin{align*}
    & J(Q^{\alpha\beta\delta}; D)
    =-\sum_{\xi_A, \xi_B}(\Pbb(\xi_A,\xi_B,0)Q^{\alpha\beta\delta}(u_A^0,u_B^0|\xi_A,\xi_B)\\
    &\qquad\qquad +\chi\Pbb(\xi_A,\xi_B,1)Q^{\alpha\beta\delta}(u_A^0,u_B^1|\xi_A,\xi_B))\\
    &=-\frac{1}{2}\sum_{\xi_A, \xi_B}(\Pbb(\xi_A, \xi_B,0)(\sim\xi_A\cdot\xi_B\oplus\alpha\cdot\xi_A\oplus\beta\cdot\xi_B\oplus\delta)
    \\
    &\qquad\qquad+\chi\Pbb(\xi_A, \xi_B, 1)(\sim\xi_A\cdot\xi_B\oplus\alpha\cdot\xi_A\oplus\beta\cdot\xi_B\oplus\delta))\\
    &\geq \frac{1}{2}\left(J(\pi_{\hat{\gamma}};D)+J(\pi_{\overline{\gamma}};D)\right),
\end{align*}
where in the last step we have used that the terms multiplying the probabilities are nonnegative. Since the RHS is independent of the no-signalling vertex, the cost of every no-signalling vertex is bounder below by the cost of a deterministic policy
$J(Q^{\alpha\beta\delta}; D)\geq \min(J(\pi_{\hat{\gamma}};D), J(\pi_{\overline{\gamma}};D)). $
Since the instance $D$ was arbitrary, this establishes that $\Cscr_{11}^a$ does not admit a no-signalling and hence quantum advantage.

We follow a similar line of arguments for $\Cscr_{11}^c.$ For an instance $D\in\Cscr(M_c, N_c)$, we have $\ell(u_A^0,u_B^0,0)=-1$, $\ell(u_A^0, u_B^1, 1)=-\chi$ and $\ell(.)\equiv 0$ otherwise. Now consider deterministic policies $\hat{\gamma}$: $\hat{\gamma}_A(\xi_A)\equiv u_A^0, \hat{\gamma}_B(\xi_B)\equiv u_B^0$ and $\overline{\gamma}$: $\overline{\gamma}_A(\xi_A)\equiv u_A^0,\overline{\gamma}_B(\xi_B)\equiv u_B^1$
It is straightforward to evaluate $J(\pi_{\hat{\gamma}};D)=-\Pbb(\xi_{W}=0)\quad
J(\pi_{\overline{\gamma}};D)=-\chi \Pbb(\xi_w=1).$

Now for any no-signalling vertex $Q^{\alpha\beta\delta}\in\Nscr\Sscr$, we again have, $J(Q^{\alpha\beta\delta}; D)\geq \half(J(\pi_{\hat{\gamma}};D)+J(\pi_{\overline{\gamma}};D))$,
whereby $J^*_{\Nscr\Sscr}(D)=J^*_\Lscr(D),$ and that $\Cscr_{11}^c$ does not admit a quantum advantage.
This establishes the proposition.
\end{proof}

\subsubsection{Elimination of 1-3 and 3-1 problem classes $\Cscr_{13}$ and $\Cscr_{31}$}
Now consider the set of $1$-$3$ class $\Cscr_{13}$. We argue that it does not admit a quantum advantage. Since $\Cscr_{31}$ is can generated by from $\Cscr_{13}$ by action $\Esf$, thanks to Proposition~\ref{prop:exchange} we need not discuss $\Cscr_{31}$ separately.
 Again, define
$\Cscr_{13}^o=\Cscr_{13}\cap\Cscr^o$,
define the achiral pair $V_a=(M_a, N_a)$ as below,
\begin{equation}\label{eq:13ma}
M_a:=\begin{pmatrix}
-1&0\\
0&0
\end{pmatrix}, N_a:=\begin{pmatrix}
0&-1\\
-1&-1
\end{pmatrix},
\end{equation}
and let, $\Cscr_{13}^a:=\Cscr( (V^a, \Omega))\cap \Cscr_{13}$ and note that
\begin{equation}\label{eq:13ca}
\Cscr_{13}^a=\{\Cscr(V_a), \Cscr(\Rsf V_a), \Cscr(V_a\Rsf), \Cscr(\Rsf V_a\Rsf)\}.
\end{equation}
Notice that we do not have a chiral class in $\Cscr_{13}$, since in such a class the matrix $N$ would have $3$ entries as $-1$, none of which overlap with the $-1$ entry in $M$, and must not be of the form in \eqref{eq:13ma}.
Notice that $|\Cscr_{13}|=\prescript{4}{}C_1\prescript{4}{}C_3=16$, $|\Cscr^o_{13}|=12$ from \eqref{eq:overlap}, and   $|\Cscr_{13}^a|=4=|\Cscr_{13}|-|\Cscr_{13}^o|$ so that $\Cscr^a_{13}$ and $\Cscr_{13}^o$ indeed themselves partition the set of all $1$-$3$ classes.
\begin{proposition}
$\Cscr_{13}$ and $\Cscr_{31}$ do not admit a quantum advantage.
\label{prop:eliminatecomp}.
\end{proposition}
\begin{proofarg}
   [of Proposition \ref{prop:eliminatecomp}]\label{app:eliminatecomp}
$\Cscr^o_{13}$ does not admit quantum advantage since $\Cscr_{13}^o\subset \Cscr^o$. Now to show the same for $\Cscr^a_{13}$, consider an instance $D\in \Cscr(M_a, N_a)$ for $V_a=(M_a, N_a)$ as in \eqref{eq:13ma}. Recall from \eqref{eq:localvertex} (from the main paper) that $\pi^{abcd}$ denotes a local deterministic strategy for all binary $a, b, c, d$.  For any no-signalling vertex $Q^{\alpha\beta\delta}$, we claim that
\begin{equation}\label{eq:13nosigdec}
    J(Q^{\alpha\beta\delta}; D)=\frac{1}{2}(J(\pi^{xyzw}; D)+J(\pi^{11ab};D))
\end{equation}
where $\pi^{xyzw}$ and $\pi^{11ab}$, as defined in \eqref{eq:localvertex} (from the main paper), are vertices of the local polytope specified by the Boolean variables $x,y,z,w,a,b\in\{0, 1\}$, which in turn are given  in terms of $\alpha, \beta$ and $\delta$ as,
\begin{align}\label{eq:x}
    x&=(\sim\beta\cdot\sim\delta)\vee(\beta\cdot\alpha\cdot\delta), \\
    y&=\sim\alpha\cdot\sim\delta\cdot\beta,\  z=\delta\vee(\sim\delta\cdot\alpha\cdot\beta)\label{eq:z}\\
    w&=(\alpha\cdot(\beta\oplus\delta))\vee(\sim\alpha\cdot\delta)\label{eq:w}\\
    a&=\sim\beta\vee(\beta\cdot\sim\alpha\cdot\sim\delta)\label{eq:a}\\
    b&=(\sim\alpha\cdot(\beta\vee\delta))\vee(\alpha\cdot(\sim\beta\oplus\delta))\label{eq:b}.
\end{align}
To establish this  claim, notice for $D\in\Cscr(V_a)$, and any $Q\in \Pscr(\Uscr|\Xi)$,
\begin{align}
    J(Q; D)&=-\sum_{\xi_A, \xi_b}\left (\Pbb(\xi_A, \xi_B, 0)Q(u_A^0,u_B^0|\xi_A, \xi_B)\nonumber\right.\\
    &\left.-\chi\Pbb(\xi_A, \xi_B, 1)(1-Q(u_A^0,u_B^0|\xi_A, \xi_B))\right ).\label{eq:1-3nscost}
\end{align}
To establish \eqref{eq:13nosigdec},  we  show that for all $\alpha, \beta, \delta, \xi_A, \xi_B \in\{0, 1\}$, and with $x,y,z,w,a,b$ as specified by \eqref{eq:x}-\eqref{eq:b},
\begin{align}\label{eq:pyclaim1}
    &Q^{\alpha\beta\delta}(u_A^0, u_B^0|\xi_A, \xi_B)\nonumber\\
    &=\frac{1}{2}(\pi^{xyzw}(u_A^0, u_B^0|\xi_A, \xi_B)+\pi^{11ab}(u_A^0, u_B^0|\xi_A, \xi_B)),
\end{align}
so that \eqref{eq:13nosigdec} now follows from \eqref{eq:pyclaim1}.
The validity of \eqref{eq:pyclaim1} can be verified through straightforward computation; due to the large number of variables involved, we relegate this to a Python notebook~\cite{c13pylink}. This establishes our claim
\eqref{eq:13nosigdec}.
We have thus established that for each no-signalling $Q^{\alpha\beta\delta}$ policy, there exists a policy in $\pi \in \Lscr$ such that $J(Q^{\alpha\beta\delta};D)\geq J(\pi;D),$ whereby establishing that $J^*_{\Nscr\Sscr}(D) \geq J^*_\Lscr(D).$ Since $D$ is arbitrary, there is no quantum advantage in $\Cscr_{13},$ and from Proposition~\ref{prop:exchange}, none in $\Cscr_{31}.$
\end{proofarg}

\subsection{$\Cscr_{12}$, $\Cscr_{21}$ and the \textit{$\half$-CAC} problem class}
We now come to the $1$-$2$ class
$\Cscr_{12}$; we will quickly address $\Cscr_{21}$ at the end of this subsection. Define $\Cscr_{12}^o=\Cscr_{12}\cap\Cscr^o,$
the achiral pair $V_a=(M_a, N_a)$,
\begin{equation}\label{eq:12ma}
 M_a:=\begin{pmatrix}
-1&0\\
0&0
\end{pmatrix}, N_a:=\begin{pmatrix}
0&-1\\
0&-1
\end{pmatrix},
\end{equation}
and the chiral pair $V_c=(M_c, N_c)$
\begin{equation}\label{eq:12mc}
	M_c:=\begin{pmatrix}
		-1&0\\
		0&0
	\end{pmatrix}, N_c:=\begin{pmatrix}
		0&-1\\
		-1&0
	\end{pmatrix}.
\end{equation}
Note that the chiral pair is $\half$-CAC form. Let $\Cscr_{12}^a:=\Cscr( (V^a, \Omega))\cap \Cscr_{12}$ and $\Cscr_{12}^c:=\Cscr( (V_c, \Omega))\cap \Cscr_{12}$.
Note that
\begin{align}\label{eq:c12adef}
 \Cscr_{12}^a=&\{\Cscr(V_a), \Cscr(\Rsf V_a),
 \Cscr(V_a\Rsf), \Cscr(\Rsf V_a\Rsf), \Cscr(\Tsf V_a), \nonumber\\
 &\quad \Cscr(\Rsf\Tsf V_a),
 \Cscr(\Tsf V_a\Rsf), \Cscr(\Rsf \Tsf V_a\Rsf)\}, \\
\Cscr_{12}^c =&\{\Cscr(V_c), \Cscr(\Rsf V_c), \Cscr(V_c\Rsf), \Cscr(\Rsf V_c\Rsf)\}.  \label{eq:c12cdef}
\end{align}
Further, notice that $|\Cscr_{12}|=\prescript{4}{}C_1\prescript{4}{}C_2=24$,
 $|\Cscr_{12}^o|+|\Cscr_{12}^a|+|\Cscr_{12}^c|=\Cscr_{12}=24$ so $\Cscr_{12}^o$, $\Cscr_{12}^a$ and $\Cscr_{12}^c$ partition the set $\Cscr_{12}$.
We eliminate all 1-2 classes not in the orbit of the $\half$-CAC class (\ie, not in the orbit of the chiral pair $(M_c,N_c)$) in the following proposition.
\begin{proposition}\label{prop:1-2Celiminate}
1) $\Cscr_{12}^o$ does not admit a quantum advantage.
2) $\Cscr_{12}^a$ does not admit quantum advantage.
\end{proposition}
\begin{proof}
1) Immediate from $\Cscr_{12}^o\subset \Cscr^o$.\\
2) For an instance $D\in\Cscr(V_a)$, Consider two deterministic policies $\hat\gamma$ and $\overline \gamma$, and the corresponding costs:
\begin{align}
    \hat{\gamma}_A(\xi_A)&\equiv u_A^0, \hat{\gamma}_B(\xi_B)\equiv u_B^0; J(\pi_{\hat{\gamma}};D)=-\Pbb(\xi_W=0)\nonumber \\
    \overline{\gamma}_A(\xi_A)&\equiv 1,\overline{\gamma}_B(\xi_B)\equiv 1;J(\pi_{\overline{\gamma}};D)=-\chi \Pbb(\xi_W=1).\label{eq:elidet}
\end{align}
Now consider for a no-signalling vertex $Q^{\alpha\beta\delta}\in\Nscr\Sscr$ and recall \eqref{eq:nonlocalvertex} (from the main paper) to express
\begin{align*}
    &J(Q^{\alpha\beta\delta}; D)=-\sum_{\xi_A, \xi_B} \Pbb(\xi_A,\xi_B,0)Q^{\alpha\beta\delta}(u_A^0,u_B^0|\xi_A,\xi_B)+\\
    & \chi\Pbb(\xi_A,\xi_B,1)\sum_{u_A} Q^{\alpha\beta\delta}(u_A, u_B^1|\xi_A, \xi_B)\\
    &=\frac{-1}{2}\sum_{\xi_A, \xi_B}(\Pbb(\xi_A, \xi_B,0)(\sim\xi_A\cdot\xi_B\oplus\alpha\cdot\xi_A\oplus\beta\cdot\xi_B\oplus\delta)\\
    &\qquad\qquad+\chi\Pbb(\xi_A, \xi_B, 1)(\xi_A\cdot\xi_B\oplus\alpha\cdot\xi_A\oplus\beta\cdot\xi_B\oplus\delta\\
    &\qquad\qquad + \sim\xi_A\cdot\xi_B\oplus\alpha\cdot\xi_A\oplus\beta\cdot\xi_B\oplus\delta)),\\
    &\geq \frac{1}{2}(J(\pi_{\hat{\gamma}};D)+J(\pi_{\overline{\gamma}};D)).
\end{align*}
In the last inequality we have again used the nonnegativity of the terms multiplying the probabilities.
Thus, the cost of every no-signalling policy is bounded below by the cost of a deterministic policy in $\Lscr$. Arguing as in Proposition~\ref{prop:11eliminate}, we see that there is no no-signalling advantage and quantum advantage within in $\Cscr^a_{12}$. This establishes the proposition.
\end{proof}

Now notice that $\Cscr_{21}=\Esf\Cscr_{12}$. Thus define $\Cscr_{21}^o=\Esf\Cscr_{12}^o, \Cscr_{21}^a=\Esf\Cscr_{12}^a$ and $\Cscr_{21}^c=\Esf\Cscr_{12}^c$, and the 2-1 class partitions into $\Cscr_{21}^o$, $\Cscr_{21}^a$ and $\Cscr_{21}^c$.
$\Cscr_{21}^c$ here lies within the orbit of the $\half$-CAC class, and the elimination of the other two $\Cscr_{21}^o$ and $\Cscr_{21}^a$ follows from Proposition \ref{prop:1-2Celiminate} and Proposition~\ref{prop:11eliminate}. This subsection thus eliminates all 1-2 and 2-1 classes that do not lie in the orbit of $\half$-CAC class.

\subsection{$\Cscr_{22}$ and the CAC problem class}
Ultimately, we attend the set of 2-2 classes  $\Cscr_{22}$. Define
$\Cscr_{22}^o=\Cscr_{22}\cap\Cscr^o$,
the chiral pair $V_c=(M_c, N_c)$ and the achiral pair $V_a=(M_a, N_a)$
\begin{equation}\label{eq:22mc}
M_c:=\begin{pmatrix}
-1&0\\
0&-1
\end{pmatrix}, N_c:=\begin{pmatrix}
0&-1\\
-1&0
\end{pmatrix}.
\end{equation}
\begin{equation}\label{eq:22ma}
 M_a:=\begin{pmatrix}
-1&0\\
-1&0
\end{pmatrix}, N_a:=\begin{pmatrix}
0&-1\\
0&-1
\end{pmatrix}.
\end{equation}
The chiral pair is in the CAC form.
Define $\Cscr_{22}^c:= \Cscr((V^c, \Omega))\cap \Cscr_{22}$ and $ \Cscr_{22}^a:=\Cscr( (V^a, \Omega))\cap\Cscr_{22}$. Notice that,
\begin{align}
\Cscr_{22}^c&=\{\Cscr(V_c), \Cscr(\Rsf V_c)\},\label{eq:c22cdef}\\
\Cscr_{22}^a&=\{\Cscr(V_a),
 \Cscr(V_a\Rsf), \Cscr(\Tsf V_a), \Cscr(\Rsf \Tsf V_a)\}. \label{eq:c22adef}
\end{align}
It is easy to check by inspection, and using \eqref{eq:overlap}, $|\Cscr^o_{22}|+|\Cscr^a_{22}|+|\Cscr^c_{22}|=|\Cscr_{22}|$
so that $\Cscr^o_{22}$, $\Cscr^a_{22}$ and $\Cscr^c_{22}$ partition $\Cscr_{22}$. Consequently, the sets $\Cscr^o_{22}$ and $\Cscr^a_{22}$ capture all 2-2 classes which are outside the orbit of CAC, and we eliminate these sets in the following proposition.
\begin{proposition}\label{prop:2-2Celiminate}
1) $\Cscr_{22}^o$ does not admit a quantum advantage.
2) $\Cscr_{22}^a$ does not admit quantum advantage.
\end{proposition}
\begin{proofarg}
[of Proposition \ref{prop:2-2Celiminate}]\label{app:2-2Celiminate}
    1) Immediate since $\Cscr_{22}^o\subset \Cscr^o$.\\
    2) 
    Let $D$ be an instance in $\Cscr_{22}^a$ and let $\hat \gamma$ and $\overline\gamma$ be as defined in \eqref{eq:elidet}. Consider any no-signalling policy $Q^{\alpha\beta\delta}\in\Nscr\Sscr$ and notice,
    \begin{align*}
        &J(Q^{\alpha\beta\delta}; D)=-\sum_{\xi_A, \xi_B}(\Pbb(\xi_A,\xi_B,0)(Q^{\alpha\beta\delta}(u_A^0,u_B^0|\xi_A,\xi_B))\\
        &\qquad +\chi\Pbb(\xi_A,\xi_B,1)(\sum_{u_A}Q^{\alpha\beta\delta}(u_A^0,u_B^1|\xi_A,\xi_B)))
    \end{align*}
    Using \eqref{eq:nonlocalvertex} (from the main paper),
    \begin{align*}
        J(Q^{\alpha\beta\delta}; D)&=-\frac{1}{2}\sum_{\xi_A, \xi_B}(\Pbb(\xi_A, \xi_B,0)+\chi\Pbb(\xi_A, \xi_B,1))\times\\
        &(\sim\xi_A\cdot\xi_B\oplus\alpha\cdot\xi_A\oplus\beta\cdot\xi_B\oplus\delta \\
        &\qquad \qquad+ \xi_A\cdot\xi_B\oplus\alpha\cdot\xi_A\oplus\beta\cdot\xi_B\oplus\delta)\\
        &=\frac{1}{2}(J(\pi_{\hat{\gamma}};D)+J(\pi_{\overline{\gamma}};D)).
    \end{align*}
    Arguing as in Proposition~\ref{prop:11eliminate}, $J^*_{\Nscr\Sscr}(D)=J^*_\Lscr(D) $,
    and the no-signalling and quantum advantages are absent in $\Cscr_{22}^a$. This establishes the proposition.
\end{proofarg}

\subsection{Optimal deterministic policy for the $\half$-CAC instance in Section \ref{sec:1/2cacpoc}}
\begin{lemma}\label{lem:1/2cacdemodet}
    For the instance $D=(M_c, N_c, \Pbb, \Uscr_A, \Uscr_B, \chi)\in \half$-CAC with $\Pbb$ given by \eqref{eq:1/2cacprior} (from the main paper) and $\chi=2$, specified in Section \ref{sec:1/2cacpoc}, an optimal deterministic policy and the corresponding cost is given by
    $\gamma^*_A\equiv u_A^0, \gamma_B^*\equiv u_B^1, \quad J(\pi_{\gamma^*};D)=-6/5.$
\end{lemma}
\begin{proof}
    We show this by directly enumerating all sixteen policies and their cost in the table below.
    \begin{center}
        \begin{tabular}{c|c|c|c|c}
            $\gamma_A(0)$ & $\gamma_A(1)$ & $\gamma_B(0)$ & $\gamma_B(1)$ & $J(\pi_{\gamma}; D)$  \\
            \hline
            $u_A^0$&$u_A^0$&$u_B^0$&$u_B^0$& $-2/5$\\
            $u_A^1$&$u_A^0$&$u_B^0$&$u_B^0$& $-6/5$\\
            $u_A^0$&$u_A^1$&$u_B^0$&$u_B^0$& $-2/5$\\
            $u_A^1$&$u_A^1$&$u_B^0$&$u_B^0$& $-6/5$\\
            $u_A^0$&$u_A^0$&$u_B^1$&$u_B^0$& $-6/5$\\
            $u_A^1$&$u_A^0$&$u_B^1$&$u_B^0$& $-6/5$\\
            $u_A^0$&$u_A^1$&$u_B^1$&$u_B^0$& $-2/5$\\
            $u_A^1$&$u_A^1$&$u_B^1$&$u_B^0$& $-2/5$\\
            $u_A^0$&$u_A^0$&$u_B^0$&$u_B^1$& $-2/5$\\
            $u_A^0$&$u_A^1$&$u_B^0$&$u_B^1$& $-4/5$\\
            $u_A^1$&$u_A^0$&$u_B^0$&$u_B^1$& $-2/5$\\
            $u_A^1$&$u_A^1$&$u_B^0$&$u_B^1$& $-4/5$\\
            $\boxed{u_A^0}$&$\boxed{u_A^0}$&$\boxed{u_B^1}$&$\boxed{u_B^1}$& $\boxed{-6/5}$\\
            $u_A^0$&$u_A^1$&$u_B^1$&$u_B^1$& $-4/5$\\
            $u_A^1$&$u_A^0$&$u_B^1$&$u_B^1$& $-2/5$\\
            $u_A^1$&$u_A^1$&$u_B^1$&$u_B^1$& $0$\\
            \hline
        \end{tabular}
    \end{center}
    It is evident that the boxed policy is an optimal policy.
\end{proof}

\section{Proof of Theorem \ref{thm:qubitsenough}}\label{proof:qubitsenough}
Our proof borrows assistance from the following two results.
The former is a lemma in linear algebra due to Jordan, which has been proved in \cite{masanes2006asymptotic} (see lemma on page 2). The latter is an embedding result for quantum strategies.
\begin{lemma}\label{prop:jordanlemma}
    \textit{(Jordan's lemma)}
    Let $P^0_0$, $P^0_1$, $P^1_0$ and $P^1_1$ be projection operators on an $n$ dimensional Hilbert space $\Hscr$ such that
    $P^0_0+P^0_1=P^1_0+P^1_1=\Ibf.$
    Then $P^0_0$, $P^0_1$, $P^1_0$ and $P^1_1$ are simultaneously block diagonalisable with each block of size at most two.
    \ie, there exists an $n_0<n$ and an orthonormal basis
    $\{\ket{w_j}, \ket{w_{-j}}\}_{j=1}^{n_0}\cup \{\ket{w_j}\}_{j=2n_0+1}^n$ such that for each $S\in\{P^0_0, P^0_1, P^1_0 ,P^1_1\}$, there exist coefficients $\{s_{qr}\}$ such that
    \begin{equation*}
        S=  \sum_{j=1}^{n_0}\sum_{q, r\in \{-j, j\}} s_{qr}\ket{w_q}\bra{w_r}+ \sum_{j=2n_0+1}^n s_{jj}\ket{w_j}\bra{w_j}.
    \end{equation*}
\end{lemma}
\begin{proposition}[An embedding for quantum strategies]\label{prop:embed}
    Let $Q=(\Hscr_A, \Hscr_B, \rho_{AB}, \{P_{u_A}^{(A)}(\xi_A)\}, \{P_{u_B}^{(B)}(\xi_B)\})$ be such that $\dim(\Hscr_{A})\leq m$ and $\dim(\Hscr_B)\leq n$. Then there exists a $\Qscr\ni Q'=(\Hscr'_A, \Hscr'_B, \rho'_{AB}, \{P_{u_A}^{(A)\prime}(\xi_A)\}, \{P_{u_B}^{(B)\prime}(\xi_B)\})$ such that
    $\dim(\Hscr_A')=m, \dim(\Hscr_B')=n$ and
    $\Tr(P_{u_A}^{(A)\prime}(\xi_A)\otimes P_{u_B}^{(B)\prime}(\xi_B)\rho'_{AB})=\Tr(P_{u_A}^{(A)}(\xi_A)\otimes P_{u_B}^{(B)}(\xi_B)\rho_{AB})$.
\end{proposition}
\begin{proof}
    Let $\dim (\Hscr_A)=m_0\leq m$ and consider an orthonormal basis $\{\ket{w_i}\}_{i=1}^{m}$ spanning $\Hscr_A'$ with $\{\ket{w_i}\}_{i=1}^{m_0}$ spanning $\Hscr_A$. Similarly let $\dim (\Hscr_B)=n_0\leq n$ and consider an orthonormal basis $\{\ket{v_j}\}_{j=1}^{n}$ spanning $\Hscr_B'$ with $\{\ket{w_j}\}_{j=1}^{n_0}$ spanning $\Hscr_B$. Let $\rho'_{AB}\in\Hscr_A'\otimes\Hscr_B'$ be such that $$\braket{w_i, v_j|\rho'_{AB}|w_k, v_\ell}=\braket{w_i, v_j|\rho_{AB}|w_k, v_\ell}$$ for $i, k\leq m_0, j, \ell \leq n_0$ and $\braket{w_i, v_j|\rho'_{AB}|w_k, v_\ell}=0$ otherwise. Such a $\rho'_{AB}$ exists by construction. To do so, one takes the submatrix of $r_{AB}$ of $\rho_{AB}$ with rows and columns corresponding to $\{\ket{w_i, v_j}\}_{i=1, j=1}^{m_0, n_0}$ and take $\rho'_{AB}=r_{AB}/\Tr{r_{AB}}$.  Since $\Hscr_i\subset \Hscr_i'$ for $i=A, B$, projectors $P_{u_i}^{(i)}(\xi_i)$ in $\Hscr_i$ are also valid projectors in $\Hscr_i'$ so let $P_{u_i}^{(i)\prime}(\xi_i)=P_{u_i}^{(i)}(\xi_i)$. It is now easy to see that $\Tr(P_{u_A}^{(A)\prime}(\xi_A)\otimes P_{u_B}^{(B)\prime}(\xi_B)\rho'_{AB})=\Tr(P_{u_A}^{(A)}(\xi_A)\otimes P_{u_B}^{(B)}(\xi_B)\rho_{AB})$.
\end{proof}

Now, we begin the proof of Theorem \ref{thm:qubitsenough}. \\
Step 1) We will show that any strategy $Q\in\Qscr$ can be written as a convex combination $\sum_{j,k} a_{jk}  Q^{jk}$, where $ Q^{jk}\in \Qscr_2$ and $a_{jk}\geq 0, \sum_{j,k}a_{jk}=1$.

Let $Q=(\Hscr_A, \Hscr_B, \rho_{AB}, \{P_{u_A}^{(A)}(\xi_A)\}, \{P_{u_B}^{(B)}(\xi_B)\})$ where $\Hscr_A$ and $\Hscr_B$ are $n$ and $m$ dimensional Hilbert spaces respectively, $P_{u_i}^i(\xi_i)$ are projection operators in $\Hscr_i$ for each $\xi_i\in\{0, 1\}, i\in\{A, B\}$ and they satisfy $P^{(i)}_{u_i^0}(0)+P^{(i)}_{u_i^1}(0)=\Ibf$ and $P^{(i)}_{u_i^0}(1)+P^{(i)}_{u_i^1}(1)=\Ibf$. Define $\Hscr_{AB}=\Hscr_A\otimes\Hscr_B$. We employ Lemma \ref{prop:jordanlemma} with $P^\xi_i \equiv P^{(A)}_{u^i_A}(\xi)$. There exists a basis $ \{\ket{w_j}, \ket{w_{-j}}\}_{j=1}^{n_0}\cup \{\ket{w_j}\}_{j=2n_0+1}^n$ and coefficients $p_{ij}(u_A, \xi_A)$ such that with the specification $p_{j-j}(u_A,\xi_A)=p_{-jj}(u_A, \xi_A)=p_{-j-j}(u_A, \xi_A) \equiv 0$ for $j>n_0$. Define $\ket{w_{-j}}=0$ for $j>2n_0$ and $P_{u_A}^{(A(j))}(\xi_A):=$
\begin{align}
    &p_{jj}(u_A, \xi_A)\ket{w_j}\bra{w_{j}}+p_{j-j}(u_A, \xi_A)\ket{w_j}\bra{w_{-j}}\label{eq:decpro1}\\
    &+p_{-jj}(u_A, \xi_A)\ket{w_{-j}}\bra{w_{j}}+p_{-j-j}(u_A, \xi_A)\ket{w_{-j}}\bra{w_{-j}}.\nonumber
\end{align}
Now for $N_0:=\{1,..., n_0\}\cup \{2n_0+1,...,n\}$, we express
\begin{equation}\label{eq:jordanpro1}
    P_{u_A}^{(A)}(\xi_A)= \sum_{j\in N_0} P_{u_A}^{A(j)}(\xi_A).
\end{equation}
Correspondingly for $m$ dimensional $\Hscr_B$,  one has a similar block-diagonalizing orthonormal basis $\{\ket{v_k}, \ket{v_{-k}}\}_{k=1}^{m_0}\cup \{\ket{v_k}\}_{k=2m_0+1}^{m}$ so that there exist coefficients $q_{ij}(u_B, \xi_B)$ such that from Lemma \ref{prop:jordanlemma}, with $q_{k-k}(u_B,\xi_B)=q_{-kk}(u_B, \xi_B)=q_{-k-k}(u_B, \xi_B)\equiv 0$ for $k>m_0$. Now let $\ket{v_{-k}}=0$ for $k>2m_0$ and
\begin{align}
    &P_{u_B}^{B(k)}(\xi_B):=q_{kk}(u_B, \xi_B)\ket{v_k}\bra{v_{k}}+q_{k-k}(u_B, \xi_B)\ket{v_k}\bra{v_{-k}}\nonumber\\
    &+q_{-kk}(u_B, \xi_B)\ket{v_{-k}}\bra{v_{k}}+q_{-k-k}(u_B, \xi_B)\ket{v_{-k}}\bra{v_{-k}}.\label{eq:decpro2}
\end{align}
Defining $M_0:=\{1,..., m_0\}\cup \{2m_0+1,...,m\}$ allows us the expression
\begin{equation}\label{eq:jordanpro2}
    P_{u_B}^{(B)}(\xi_B)=\sum_{k\in M_0} P_{u_B}^{B(k)}(\xi_B).
\end{equation}
$\rho_{AB}$ can be expressed as
\begin{equation}\label{eq:jordanrho}
    \rho_{AB}=\sum_{j\in N_0}\sum_{k\in M_0} r_{AB}^{jk}
\end{equation}
where $r_{AB}^{jk}$ denotes an operator of the form
$$r_{AB}^{jk}=\sum_{i, p=\pm j}\sum_{l, q=\pm k} r(i, l, p, q)\ket{w_i, v_l}\bra{w_p, v_q}.$$
Since $\rho_{AB}\succeq 0$ we have $\Tr (r^{jk}_{AB})\geq 0$ for all $j\in N_0, k\in M_0$  as we have $$\braket{w_{\pm j}, v_{\pm k}|\rho_{AB}| w_{\pm j}, v_{\pm k}}=\braket{w_{\pm j}, v_{\pm k}|r^{jk}_{AB}| w_{\pm j}, v_{\pm k}}\geq 0.$$
Let $\Hscr_A^j=\Span\{\ket{w_j}, \ket{w_{-j}}\}$ and $\Hscr_B^k=\Span\{\ket{v_k}, \ket{v_{-k}}\}$. Define an operator $\rho_{AB}^{jk}$ on  $\Hscr^j_A\times \Hscr_B^k$ given by
\begin{equation}\label{eq:jordanrhosimp}
    \rho_{AB}^{jk}:=\begin{cases}
        \Tr(r_{AB}^{jk})\inv r_{AB}^{jk} & r^{jk}_{AB}\neq 0\\
        0 & r^{jk}_{AB}=0
    \end{cases}
\end{equation}
We claim that $\rho_{AB}^{jk}$ is a density operator when $r^{jk}_{AB}\neq 0$. Notice that $\Tr (\rho_{AB}^{jk}) =1$ by definition and $r^*(p,q,i,l)=r(i,l,p,q)$ for all $i, l, p, q$ since $\rho_{AB}$ in \eqref{eq:jordanrho} is Hermitian so that $\rho^{jk\dagger}_{AB}=\rho^{jk}_{AB}$. Finally, $\rho_{AB}\succeq 0$ since $r^{jk}_{AB}\succeq 0$ and $\Tr (r^{jk}_{AB})\geq 0$. Further, as defined by \eqref{eq:decpro1} {and} \eqref{eq:decpro2}, consider projectors $P_{u_A}^{A(j)}(\xi_A)$ and $P_{u_B}^{B(k)}(\xi_B)$ in spaces $\Hscr_A^j$ and $\Hscr_B^k$ respectively. Recall \eqref{eq:jordanpro1} and \eqref{eq:jordanpro2} and specify strategies $Q^{jk}:=(\Hscr_A^j, \Hscr_B^k, \rho_{AB}^{jk}, P_{u_A}^{A(j)}(\xi_A),  P_{u_B}^{B(k)}(\xi_B))$ for  $j\in N_0,k\in M_0$. Now $1\leq \dim(\Hscr_A^j)=\dim(\Hscr^k_B)\leq 2$ by construction so from Proposition, we have $Q^{jk}\in\Qscr_2$, and notice from \eqref{eq:jordanpro1}, \eqref{eq:jordanpro2}, \eqref{eq:jordanrho} and \eqref{eq:jordanrhosimp} that
\begin{align*}
    Q(u_A, u_B&|\xi_A, \xi_B)=\Tr(P_{u_A}^{(A)}(\xi_A)\otimes P_{u_B}^{(B)}(\xi_B) \rho_{AB})\\
    &=\sum_{j,k} \Tr(P_{u_A}^{A(j)}(\xi_A)\otimes P_{u_B}^{B(k)}(\xi_B) r^{jk}_{AB})\\
    &=\sum_{j,k} \Tr(r^{jk}_{AB}) \Tr(P_{u_A}^{A(j)}(\xi_A)\otimes P_{u_B}^{B(k)}(\xi_B) \rho^{jk}_{AB})\\
    &=\sum_{j,k} \Tr(r^{jk}_{AB}) Q^{jk}(u_A, u_B|\xi_A, \xi_B).
\end{align*}
Now $\Tr(r^{jk}_{AB}) \geq 0$ and
$\sum_{jk} \Tr(r^{jk}_{AB})=\Tr(\rho_{AB})=1$ whereby $Q$ is a convex combination of strategies $ Q^{jk}\in\Qscr_2$.  This completes the proof of the first part of the theorem.\\
Step 2) First, note that the inequality $\inf_{Q\in\Qscr_2} J(Q;D)\geq \inf_{Q\in\Qscr} J(Q;D)$ holds since $\Qscr_2\subset \Qscr$. Suppose that the inequality is strict so that $\inf_{Q\in\Qscr_2} J(Q;D)= \inf_{Q\in\Qscr} J(Q;D)+\delta$ where $\delta>0$. But $\exists Q\in \Qscr$ such that $J(Q;D)\leq \inf_{Q\in\Qscr} J(Q;D)+\delta/2$ by definition of an infimum.
Expected cost is a linear function of $Q$ so that from part 1) we have $Q$, $\exists a_{jk}>0, \sum a_{jk}=1, Q^{jk}\in\Qscr_2$ such that $J(Q;D)=\sum_{j,k}a_{jk}\sum_\xi\sum_u J(u, \xi)Q^{jk}(u|\xi)$. It follows that $$\min_{j,k} J(Q^{jk};D)\leq \inf_{Q\in\Qscr} J(Q;D)+\delta/2 < \inf_{Q\in\Qscr_2} J(Q;D)$$ which is a contradiction.\\
Step 3) Let us attend our first claim: i) $\Qscr_2=\conv(\underline{\Qscr}_2)$.
Consider an arbitrary $Q=(\Hscr_A, \Hscr_B, \rho_{AB},  \{P_{u_A}^{(A)}(\xi_A)\}, \{P_{u_B}^{(B)}(\xi_B)\})\in \Qscr_2$. Then there exist finitely many states (which we index by $i$ in some finite range) $\ket{\psi_i}\in \Hscr_A\otimes \Hscr_B$ and corresponding probabilities $p_i \geq 0, \sum p_i =1$ such that
$\rho_{AB}=\sum_i p_i\ket{\psi_i}\bra{\psi_i}.$
Now define
$Q_i:=(\Hscr_A, \Hscr_B, \ket{\psi_i}\bra{\psi_i},  \{P_{u_A}^{(A)}(\xi_A)\}, \{P_{u_B}^{(B)}(\xi_B)\})\in\underline{\Qscr}_2 .$
From the linearity of trace, $Q$ satisfies
\begin{align*}
    Q(u_A, u_B|\xi_A, \xi_B)
    =\sum_{i}p_i Q_i(u_A, u_B|\xi_A, \xi_B).
\end{align*}
Hence, $Q\in\conv(\underline\Qscr_2)$ holds. The claim $\Qscr_2=\conv(\underline\Qscr_2)$ follows since $Q$ was arbitrarily chosen. \\
Towards ii) $\underline\Qscr_2=\widehat\Qscr_2$, consider any $Q\in\underline{\Qscr}_2$:
\begin{equation*}
    Q=(\Hscr_A, \Hscr_B, \ket{\psi_{AB}}\bra{\psi_{AB}},
    \{P_{u_A}^{(A)}(\xi_A)\}, \{P_{u_B}^{(B)}(\xi_B)\})
\end{equation*}
such that $\ket{\psi_{AB}}\in \Hscr_A\otimes \Hscr_B$. Let $\{\ket{z^+_i}, \ket{z^-_i}\}$ denote an orthonormal basis of $\Hscr_i$ for $i=A, B$. We can then express $\ket{\psi_{AB}}\in \Hscr_A\otimes \Hscr_B$ as a unit vector in the basis $\{\ket{z^+_A, z^+_B}, \ket{z^+_A, z^-_B}, \ket{z^-_A, z^+_B}, \ket{z^-_A, z^-_B}\}$. Let such an expression be
$$\ket{\psi_{AB}}=c_{1}\ket{z^+_A, z^+_B}+ c_{2}\ket{z^+_A, z^-_B}+c_{3}\ket{z^-_A, z^+_B}+c_{4}\ket{z^-_A, z^-_B}$$
for some coefficients $c_{1}, c_{2}, c_{3}, c_{4}\in \mathbb{C}$. It is well known that each such $\ket{\psi_{AB}}$ admits a Schmidt decomposition \cite{neilsen2004qcqi}, \ie, for each (normalized) tuple $(c_{1}, c_{2}, c_{3}, c_{4})$ there exists an $ a\in [0, 1]$ and orthonormal bases (also known as the Schmidt bases), $\{\ket{s^+_A}, \ket{s^-_A}\}$ and  $\{\ket{t^+_B}, \ket{t^-_B}\}$ of  $\mathcal{H}_A$ and $\mathcal{H}_B$ such that
\begin{equation}
    \ket{\psi_{AB}}:=a\ket{s^+_A, t^+_B } + \sqrt{1-a^2}\ket{s^-_A, t^-_B}.
\end{equation}
Take $\alpha=2\arccos(a)$. Further, we can find unitary operators $\sf{U}_i\in\Bscr(\Hscr_i)$, $i\in\{A, B\}$ that provide transformations among the given basis vectors $\{\ket{z^+_i}, \ket{z^-_i}\}, i=A, B$ and the Schmidt bases as follows: $\mathsf{U}_A\ket{z^+_A}=\ket{t^+_A}, \mathsf{U}_A\ket{z^-_A}=\ket{t^-_A}, \mathsf{U}_B\ket{z^+_B}=\ket{s^+_B}, \mathsf{U}_B\ket{z^-_B}=\ket{s^-_A}.$
Notice that then
$\mathsf{U}_A\otimes \mathsf{U}_B \ket{\Phi^{\alpha}_{AB}}=\ket{\psi_{AB}}$ where we recall our notation $\ket{\Phi^\alpha_{AB}}$ from \eqref{eq:alphawavefunc}.
Define
$$K_{u_A}^{(A)}(\xi_A)=\mathsf{U}_A^\dagger P_{u_A}^{(A)}(\xi_A)\mathsf{U}_A,
K_{u_B}^{(B)}(\xi_B)=\mathsf{U}_B^\dagger P_{u_B}^{(B)}(\xi_B)\mathsf{U}_B,$$
and notice that $$K_{u_i}^{(i)}(\xi_i)^\dagger=K_{u_i}^{(i)}(\xi_i), K_{u_i}^{(i)}(\xi_i)K_{u_i}^{(i)}(\xi_i)=K_{u_i}^{(i)}(\xi_i)$$ so that they are valid projection operators. Define, \\
$\widehat{Q}:=(\Hscr_A, \Hscr_B, \ket{\Phi^{\alpha}_{AB}}\bra{\Phi^{\alpha}_{AB}}, \{K_{u_A}^{(A)}(\xi_A)\}, \{K_{u_B}^{(B)}(\xi_B)\}).$
Clearly $\widehat{Q}\in \Qscr^\alpha_2$.  Further notice that  $$\widehat Q(u_A, u_B|\xi_A, \xi_B)=\Tr(K_{u_A}^{(A)}(\xi_A)\otimes K_{u_B}^{(B)}(\xi_B) \ket{\Phi^{\alpha}_{AB}}\bra{\Phi^{\alpha}_{AB}}).$$ Substituting $K_{u_A}^{(A)}(\xi_A)$ and $K_{u_B}^{(B)}(\xi_B)$ and using the unitarity of $\mathsf{U}_A, \mathsf U_B$, it follows that
\begin{align*}
    \widehat Q(u_A, u_B|\xi_A, \xi_B)&=
    \Tr(P_{u_A}^{(A)}(\xi_A)\otimes P_{u_B}^{(B)}(\xi_B)\ket{\psi_{AB}}\bra{\psi_{AB}})\\ &= Q (u_A, u_B|\xi_A, \xi_B).
\end{align*}
The claimed equality $\underline\Qscr_2=\widehat\Qscr_2$ follows since $Q$ was arbitrary.\\
4) Recall that $\inf_{Q\in\Qscr} J(Q;D)=\inf_{Q\in\Qscr_2} J(Q;D)$ holds from Theorem \ref{thm:qubitsenough}. On the other hand, since $\underline\Qscr_2=\widehat\Qscr_2$, $\inf_{Q\in\widehat\Qscr_2} J(Q;D)=\inf_{Q\in\underline\Qscr_2} J(Q;D)$ follows. It remains to show that $\inf_{Q\in\Qscr_2} J(Q;D)=\inf_{Q\in\underline\Qscr_2} J(Q;D)$ so suppose otherwise, \ie, $\inf_{Q\in\underline\Qscr_2} J(Q;D)-\inf_{Q\in\Qscr_2} J(Q;D)=\delta>0$.  Recall that $\Qscr_2=\conv(\underline\Qscr_2)$.
By definition of an infimum, we have a $Q\in\Qscr_2$ such that $J(Q;D)< \inf_{Q\in\Qscr_2} J(Q;D)+\delta/2$.
Now for some set of $p_i\geq 0, \sum_i p_i=1$, we have $Q=\sum_i p_i Q_i$ where $Q_i\in \underline \Qscr_2$.
It follows that $\min_i J(Q_i; D)\leq \inf_{Q\in\Qscr_2} J(Q;D)+\delta/2$ which contradicts the supposition. Thus $\inf_{Q\in\underline{\Qscr}_2}=\inf_{Q\in\Qscr_2} J(Q;D)$. We hereby conclude our proof of the Theorem \ref{thm:qubitsenough}.
%

\section{Proof of Proposition \ref{prop:alphaprob}}\label{sec:alphaprob}
We now prove Proposition \ref{prop:alphaprob} and show how Table \ref{tab:alphaprobs} (from the main paper) can be derived to paramterically specify our quantum strategies in $\Qscr^\alpha_2$ (recall \eqref{eq:alphaqubitellitope} in the main paper).
We will drop the subscripts $i(\xi_i)$ in the kets and on parameters that refer to players and observations for now, for their presence is apparent. So let
$$a:=(\cos{\theta_a}\sin{\phi_a}, \sin{\theta_a}\sin{\phi_a}, \cos{\phi_a}), $$
$$b:=(\cos{\theta_b}\sin{\phi_b}, \sin{\theta_b}\sin{\phi_b}, \cos{\phi_b}).$$
For brevity, let $Q(u_A^+, u^+_B) := Q(u^+_A(\xi_A),u^+_B(\xi_B)|\xi_A,\xi_B)$.
\begin{align*}
    Q(u_A^+, u^+_B)&=|\braket{a^+, b^+|\Phi^\alpha_{AB}}|^2\\
    =&|\cos\frac{\alpha}{2}\braket{a^+, b^+|z^+, z^+}+\sin\frac{\alpha}{2}\braket{a^+, b^+|z^-, z^-}|^2.
\end{align*}
Now using Bloch representation~\cite{neilsen2004qcqi} for qubit states, we have for $i=a, b$:
\begin{align*}
\ket{i^+}&=\cos\frac{\phi_i}{2}\ket{z^+}+e^{\iota\theta_i}\sin\frac{\phi_i}{2}, \\ \ket{i^-}&=-e^{-\iota\theta_i}\sin\frac{\phi_i}{2}\ket{z^+}+\cos\frac{\phi_i}{2}\ket{z^-}
\end{align*}
 whereby
 \begin{align*}
     \braket{i^+|z^+}&=\cos\frac{\phi_i}{2}, \braket{i^+|z^-}=e^{\iota\theta_i}\sin\frac{\phi_i}{2}, \\ \braket{i^-|z^+}&=-e^{-\iota\theta_i}\sin\frac{\phi_i}{2}, \braket{i^-|z^-}=\cos\frac{\phi_i}{2}.
 \end{align*}
Indeed then,
\begin{align*}
    Q&(u_A^+, u_B^+)=+\sin^2{(\theta_a+\theta_b)}\sin^2\frac{\alpha}{2}\sin^2\frac{\phi_a}{2}\sin^2\frac{\phi_b}{2}\\ &+|\cos\frac{\alpha}{2}\cos\frac{\phi_a}{2}\cos\frac{\phi_b}{2}+e^{-\iota(\theta_a+\theta_b)}\sin\frac{\alpha}{2}\sin\frac{\phi_a}{2}\sin\frac{\phi_b}{2}|^2\\
    &=\cos^2\frac{\alpha}{2}\cos^2\frac{\phi_a}{2}\cos^2\frac{\phi_b}{2}+\sin^2\frac{\alpha}{2}\sin^2\frac{\phi_a}{2}\sin^2\frac{\phi_b}{2}\\
    &+\frac{1}{4}\sin\alpha\cos{(\theta_a+\theta_b)}\sin{\phi_a}\sin{\phi_b}.
\end{align*}
Other entries in Table \ref{tab:alphaprobs} can be similarly computed.

\section{Proof of Proposition \ref{prop:pienough}}\label{app:pienough}
We sketch a rigorous proof below, skipping some easily verifiable algebraic details.
First, note that $\Qscr^*$ is non-empty since $\widehat \Qscr_2$ is compact since it is specified by a compact set of parameters along with a finite number of action assignments $\{u^\pm_i(\xi_i)\}$ whereby the minimization is over a finite union of compact sets. Now let $Q=(\alpha, \theta, \phi, \{u_i^\pm(\xi_i)\})\in\Qscr^*$. We now specify a quantum strategy $Q^0$ with the desired property directly and show that it attains the same cost as $Q$. Take $Q^0=(\alpha, \theta^0, \phi^0, \{u^\pm_i(\xi_i)\})$ where $\alpha$ and $\{u^\pm_i(\xi_i)\}$ are the same as in $Q$. Suppose the measurement directions in $Q,Q^0$ are denoted by $a(\xi_A), b(\xi_B)$ and $a^0(\xi_A),b^0(\xi_B)$. We rotate the physical frame so that the $x$-axis now points along $a(0)$ and we take $a^0(\xi_A),b^0(\xi_B)$ to be along $a(\xi_A),b(\xi_B)$  but described using polar coordinates with respect to this new frame. Thus,
$\theta^0_{a(0)}=\phi^0_{a(0)}=0$. The other angles are given through the equations
\begin{align}
    &i^0(\xi_i)=(\cos\theta^0_{i^0(\xi_i)}\sin\phi^0_{i^0(\xi_i)}, \sin\theta^0_{i^0(\xi_i)}\sin\phi^0_{i^0(\xi_i)}, \cos\phi^0_{i^0(\xi_i)} )\nonumber\\
    &=(\sin\phi_{i(\xi_i)}(\cos\phi_{a0}\cos\theta_{a0}\cos(\theta_{a0}-\theta_{i(\xi_i)})\nonumber\\
    &+\sin\theta_{a0}\sin(\theta_{a0}-\theta_{i(\xi_i)}))-\cos\phi_{i(\xi_i)}\cos\theta_{a0}\sin\phi_{a0}, \nonumber\\
    & \sin\phi_{i(\xi_i)}(\cos\phi_{a0}\sin\theta_{a0}\cos(\theta_{a0}-\theta_{i(\xi_i)})\nonumber\\
    &-\cos\theta_{a0}\sin(\theta_{a0}-\theta_{i(\xi_i)}))-\cos\phi_{i(\xi_i)}\sin\theta_{a0}\sin\phi_{a0}, \nonumber\\
    & \cos\phi_{a0}\cos\phi_{i(\xi_i)}+\cos(\theta_{a0}-\theta_{i(\xi_i)}\sin\phi_{a0}\sin\phi_{i(\xi_i)})\nonumber\\
    &=:(i^0(\xi_i)_x, i^0(\xi_i)_y, i^0(\xi_i)_z).\label{eq:impliciteq}
\end{align}
for $\xi_i\in\{0, 1\}, i=A, B$.
We will demonstrate how $Q(u_A^+(\xi_A), u_B^+(\xi_B)|\xi_A, \xi_B)=Q^0(u_A^+(\xi_A), u_B^+(\xi_B)|\xi_A, \xi_B)$ can be shown, following which the rest of the equalities can be similarly verified.
We know $Q^0(u_A^+(\xi_A), u_B^+(\xi_B)|\xi_A, \xi_B)=|\braket{a^0(\xi_A)^{+}, b^0(\xi_B)^{+}|\Phi^\alpha_{AB}}|^2=$ $$\Tr(\ket{\Phi^\alpha_{AB}}\bra{\Phi^\alpha_{AB}}\ket{a^0(\xi_A)^+}\bra{a^0(\xi_A)^+}\otimes \ket{b^0(\xi_B)^+}\bra{b^0(\xi_B)^+}).$$
Now, using the well known Pauli formula for spin-2 density matrices~\cite{preskill1998notes}, we have
$$\ket{a^0(\xi_A)^+}\bra{a^0(\xi_A)^+}=\begin{pmatrix}
    \frac{1+a^0(\xi_A)_z}{2} & \frac{a^0(\xi_A)_x-\iota a^0(\xi_A)_y}{2}\\
    \frac{a^0(\xi_A)_x+\iota a^0(\xi_A)_y}{2} & 	\frac{1-a^0(\xi_A)_z}{2}
\end{pmatrix}$$
$$\ket{b^0(\xi_B)^+}\bra{b^0(\xi_B)^+}=\begin{pmatrix}
    \frac{1+b^0(\xi_B)_z}{2} & \frac{b^0(\xi_B)_x-\iota b^0(\xi_B)_y}{2}\\
    \frac{b^0(\xi_B)_x+\iota b^0(\xi_B)_y}{2} & 	\frac{1-b^0(\xi_B)_z}{2}
\end{pmatrix}$$
where $a^0(\xi_A)_x, a^0(\xi_A)_y, a^0(\xi_A)_z, b^0(\xi_B)_x, b^0(\xi_B)_y$ and $b^0(\xi_B)_z$ are evaluted in terms of $\theta^0, \phi^0$ through \eqref{eq:impliciteq}. Hence substituting from \eqref{eq:impliciteq} and evaluating $\Tr(\ket{\Phi^\alpha_{AB}}\bra{\Phi^\alpha_{AB}}\ket{a^0(\xi_A)^+}\bra{a^0(\xi_A)^+}\otimes \ket{b^0(\xi_B)^+}\bra{b^0(\xi_B)^+})$ on an algebraic calculator confirms $Q^0((u_A^+(\xi_A), u_B^+(\xi_B)|\xi_A, \xi_B)$ matches $Q(u_A^+(\xi_A), u_B^+(\xi_B)|\xi_A, \xi_B)$ where the latter is available in Table \ref{tab:alphaprobs} as before.

\section{Proof of Proposition \ref{prop:bounds}}\label{app:bounds}
We have, from the simplification in Equation~\eqref{eq:costsimplify} (from the main paper)
\begin{align*}
    J(\pi; D)&=-\sum\Pbb(\xi_A, \xi_B, 0)\pi(u_A=u_B|\xi_A, \xi_B)\\
    &-\chi\sum\Pbb(\xi_A, \xi_B, 1)\pi(u_A\neq u_B|\xi_A, \xi_B)
\end{align*}
It is straightforward to compute that
\begin{align*}
    J(\pi^{0000};D)&=J(\pi^{0010};D)=-k-s-2t \\
    J(\pi^{0011};D)&=J(\pi^{0011};D)=-\chi(s+k+2t) \\
    J(\pi^{1100};D)&= J(\pi^{1110};D)=-k-s-2\chi t \\
    J(\pi^{1101};D)&= J(\pi^{1110};D)=-\chi(k+s)-2t,\\
    J(\pi;D)&=-(1+\chi)(s+t)
\end{align*}
 for $\pi\in\{\pi^{0100}, \pi^{1000},\pi^{0110},\pi^{1010}\}$ and $J(\pi;D)=-(1+\chi)(k+t)$ for $\pi\in\{\pi^{1001}, \pi^{1011},\pi^{0101},\pi^{0111}\}$.
Proposition \ref{prop:bounds} now follows trivially.

\section{Proof of Proposition \ref{prop:idassign}}\label{proof:idassign}
We will require repeated reference to Table \ref{tab:uaequaltoub} (from the main paper) throughout this proof.
\begin{table*}[t]
    \centering
    \begin{tabular}{|c|c|c|}
        \hline
        No. & Assignment Case $(j= -i)$ & $Q(u_i=u_j|\xi_i, \xi_j)$ \\
        \hline
        I & $u_i^+(\xi_i)=u_i^-(\xi_i)=u_j^+(\xi_j)=u_j^-(\xi_j)$ & $\delta(u_i, u_i^+(\xi_i))\delta(u_j, u_j^+(\xi_j))$ \\
        \hline
        II & $u_i^+(\xi_i)=u_i^-(\xi_i)=u_j^+(\xi_j)\neq u_j^-(\xi_j)$ & $
        \cos^2({\alpha}/{2})\cos^2({\phi_{j(\xi_j)}}/{2})+\sin^2({\alpha}/{2})\sin^2({\phi_{j(\xi_j)}}/{2})$ \\
        \hline
        III &  $u_i^+(\xi_i)=u_i^-(\xi_i)=u_j^-(\xi_j)\neq u_j^+(\xi_j)$& $
        \cos^2({\alpha}/{2})\sin^2({\phi_{j(\xi_j)}}/{2})+\sin^2({\alpha}/{2})\cos^2({\phi_{j(\xi_j)}}/{2})$  \\
        \hline
        IV & $u_i^+(\xi_i)\neq u_i^-(\xi_i)=u_j^-(\xi_j)\neq u_j^+(\xi_j)$ & $
        \cos^2\frac{\phi_{i(\xi_i)}}{2}\cos^2\frac{\phi_{j(\xi_j)}}{2}+\sin^2\frac{\phi_{i(\xi_i)}}{2}\sin^2\frac{\phi_{j(\xi_j)}}{2}+ 2\beta_{\alpha, \theta_{i(\xi_i)}, \theta_{j(\xi_j)}, \phi_{i(\xi_i)}, \phi_{j(\xi_j)}}$ \\
        \hline
        V & $u_i^+(\xi_i)\neq u_i^-(\xi_i)=u_j^+(\xi_j)\neq u_j^-(\xi_j)$ & $
        \cos^2\frac{\phi_{i(\xi_i)}}{2}\sin^2\frac{\phi_{j(\xi_j)}}{2}+\sin^2\frac{\phi_{i(\xi_i)}}{2}\cos^2\frac{\phi_{j(\xi_j)}}{2}- 2\beta_{\alpha, \theta_{i(\xi_i)}, \theta_{j(\xi_j)}, \phi_{i(\xi_i)}, \phi_{j(\xi_j)}}$ \\
        \hline
    \end{tabular}
    \caption{Probability of equal actions over different assignment cases.  Here $
        \beta_{\alpha, \theta_{a(\xi_A)}, \theta_{b(\xi_B)}, \phi_{a(\xi_A)}, \phi_{b(\xi_B)}}
        =(1/4)\sin\alpha\cos(\theta_{a(\xi_A)}+\theta_{b(\xi_B)})\sin\phi_{a(\xi_A)}\sin\phi_{b(\xi_B)}
        $}
    \label{tab:uaequaltoub}
\end{table*}
Let $Q=(\alpha, \theta, \phi, \{u^{\pm}_i(\xi_i)\}_i)\in\widehat{\Qscr}$. We show that for any $Q$ there exists $\alpha^0, \theta^0, \phi^0$ and $Q_0:= (\alpha^0, \theta^0, \phi^0, \{v_i^\pm(\xi_i)\})$ such that $J(Q_0;D)\leq J(Q;D)$. To do this, we run through the possible action assignments $\{u^\pm_i(\xi_i)\}_i$ in $Q$ over the cases below. Table \ref{tab:uaequaltoub} enumerates the relevant expression for $Q$ for each action assignment.

We introduce some notation we utilise extensively through this proof, and sparsely later in this article.
We denote $J(Q';D)-J(Q;D)=:\Lambda(Q',Q)$ while the instance $D$ is fixed by context. Recall from \eqref{eq:costsimplify} (from the main paper) that $\kappa(\xi_A, \xi_B)=\chi\Pbb(\xi_A, \xi_B, 1)-\Pbb(\xi_A,\xi_B, 0)$. Denote $$d(Q', Q|\xi_A, \xi_B):=Q'(u_A=u_B|\xi_A, \xi_B)-Q(u_A=u_B|\xi_A, \xi_B).$$ It is then clear from \eqref{eq:costsimplify} that
$$\Lambda(Q',Q)=\sum_{\xi_A, \xi_B}\kappa(\xi_A, \xi_B) d(Q', Q|\xi_A, \xi_B).$$  Additionally, if the assignment $u$ obeys conditions of row $X$ of Table \ref{tab:uaequaltoub}, we denote $Q(u_A=u_B|\xi_A, \xi_B)$ by $Q^X(u_A=u_B|\xi_A, \xi_B)$.  $C_i$ denotes the set of assignments encapsulated by a \textit{Case} $i$ below.

\textit{Case 1:} $C_1:=\{u|u^+_i(\xi_i)\neq u^-_i(\xi_i) \forall i\in \{A, B\}, \xi_i\in\{0, 1\}\}$.

$|C_1|=16$ (two choices for each of $u_A^+(0), u_A^+(1), u_B^+(0), u_B^+(1)$), and $\{v_i^\pm(\xi_i)\}_i$ which is particular to the our parameterised $Q_0(\alpha, \theta, \phi)$ obeys $v\in C_1$. The key strategic ingredient that goes in the cost is $Q(u_A=u_B|\xi_A, \xi_B)$ computed for each case in Table \ref{tab:uaequaltoub}. For strategies that resort to \textit{Case 1}, the relevant computation is given by the rows IV and V in the table. Clearly, the assignment $\{v_i^\pm(\xi_i)\}_i$ corresponding to $Q_0$ obeys conditions of row IV. We show that $\alpha^0, \theta^0$ and $\phi^0$ can be tuned so appropriately for $Q_0$ to match any $Q$.
So let $\alpha^0=\alpha,$
\begin{equation}\label{eq:theta0c1}
    \theta^0_{i(\xi_i)}=\begin{cases}
        \theta_{i(\xi_i)} & u_i^{0}= u_i^+(\xi_i)\\
        \pi+\theta(\xi_i)+2m(\xi_i)\pi  & \text{otherwise},
    \end{cases}
\end{equation}
where $m(\xi_i)\in\{0, -1\}$ is set to ensure $\theta'_{i(\xi_i)}\in [0, 2\pi)$ and
\begin{equation}\label{eq:phi0c1}
    \phi^0_{i(\xi_i)}=\begin{cases}
        \phi(\xi_i) & u_i^{0}= u_i^+(\xi_i)\\
        \pi-\phi(\xi_i)+2n(\xi_i)\pi  & \text{otherwise}.
    \end{cases}
\end{equation}
where similarly $n(\xi_i)\in\{0, -1\}$ is set so that $\phi'\in [0, 2\pi)$. It is now clear by an inspection of Table \ref{tab:uaequaltoub} that $Q_0^{IV}(u_A=u_B|\xi_A, \xi_B) \equiv Q^V(u_A=u_B|\xi_A, \xi_B)$ so that $J(Q_0;D)= J(Q;D)$. We quickly illustrate one such inspection for the reader's clarity. Suppose that the assignment corresponding to our arbitrary $Q$ obeys $u^+_A(\xi_A)=u_A^1$ and $u^+_B(\xi_B)=u_B^0$ for a particular tuple $\xi_A, \xi_B$ which pertains to row V. Our transformation \eqref{eq:theta0c1}, \eqref{eq:phi0c1} then ensures $\cos^2\frac{\phi_{a(\xi_A)}^0}{2}=\sin^2\frac{\phi_{a(\xi_A)}}{2}$ and $\beta(\alpha^0, \theta^0, \phi^0)=-\beta(\alpha, \theta, \phi)$ so that $Q^{IV}_{0}(u_A=u_B|\xi_A, \xi_B)=Q^V(u_A=u_B|\xi_A, \xi_B)$.  \\
\textit{Case 2:} $C_2=\{u^+_j(\xi_j^*)=u^-_j(\xi_j^*)$ for a unique pair $(j, \xi_j^*) \in\{A, B\}\times \{0, 1\}\}$.

We have $|C_2|=64$ and $v\notin C_2$. Let $Q':=(\alpha', \theta', \phi', \{w^\pm_i\}_i)\in\widehat\Qscr$ where $w\in C_1$.
Use $\alpha'=\alpha$, $w^+_j(\xi^*_A)=u^+_j(\xi^*_j)\neq w^-_j(\xi^*_j)$ and $w^\pm_i(\xi_i)\equiv u^\pm_i(\xi_i)$ for $i\neq j$.  We specialise to $j=A, i=B$ as the argument for $j=B, i=A$ will then similarly follow. We have two subcases within \textit{case 2}.
\\
\textit{Case 2a:} $w^+_B(\xi_B)=w^+_A(\xi_A^*)$ so that $Q'(u_A=u_B|\xi_A^*, \xi_B)=Q'^{IV}(u_A=u_B|\xi_A^*, \xi_B)$ and $Q(u_A=u_B|\xi_A, \xi_B)=Q^{II}(u_A=u_B|\xi_A, \xi_B)$. Set $\phi'_{a(\xi_A^*)}\equiv \alpha$ and $\phi'_{b(\xi_B)}\equiv \phi_{b(\xi_B)}$. Then,
\begin{align*}
    \Lambda(Q, Q')
    &=\sum_{\xi_A,\xi_B} \kappa(\xi_A, \xi_B) d(Q', Q|\xi_A, \xi_B)\\
    &={-\sin\alpha}/{2} \sin\phi_{a(\xi_A^*)}\times \\
    &\quad \sum_{\xi_B}\kappa(\xi_A^*, \xi_B) \sin\phi_{b(\xi_B)}\cos(\theta'_{a(\xi_A^*)}+\theta'_{b(\xi_B)}).
\end{align*}

\noindent \textit{Case 2b:} $w^-_B(\xi_B)=w^+_A(\xi_A^*)$  so that  $$Q'(u_A=u_B|\xi_A^*, \xi_B)=Q'^{V}(u_A=u_B|\xi_A^*, \xi_B)$$ and $Q(u_A=u_B|\xi_A, \xi_B)=Q^{III}(u_A=u_B|\xi_A, \xi_B)$. Similarly set $\phi'_{a(\xi_A^*)}\equiv \alpha$ and $\phi'_{b(\xi_B)}\equiv \phi_{b(\xi_B)}$. Then, $\Lambda(Q, Q')=$ $$
\frac{\sin\alpha}{2} \sin\phi_{a(\xi_A^*)}\sum_{\xi_B}\kappa(\xi_A^*, \xi_B) \sin\phi_{b(\xi_B)}\cos(\theta'_{a(\xi_A^*)}+\theta'_{b(\xi_B)}).$$
For each of the above two cases, notice that the transformation $\theta'_{a(\xi_A^*)}\to \theta'_{a(\xi_A^*)}\pm\pi$ takes $\Lambda(Q',Q)\to -\Lambda(Q',Q)$ so $\exists\ \theta'_{a(\xi_A^*)}$ such that $\Lambda(Q, Q')\leq 0$. This provides the construction of a $Q'$ for every $Q$ such that $J(Q'; D)\leq J(Q; D)$. Since $Q'$ pertains to \textit{case 1} of our proof, it follows that $\exists$ a $Q_0(\alpha^0, \theta^0, \phi^0)$ obeying $J(Q_0(\alpha^0, \theta^0, \phi^0);D)\leq J(Q';D)\leq J(Q;D)$.

\textit{Case 3:}  $C_3=\{\exists\ (\xi_A^*, \xi_B^*)\in \{0, 1\}^2$ such that $u_i^+(\xi^*_i)=u_i^-(\xi^*_i)$  and $u_i^+(\xi'_i)\neq u_i^-(\xi'_i)$ for each $i\in\{A, B\}$ where $\xi'_i:=\sim\xi_i^*\}$.

$|C_3|=64$. Consider a $Q':=(\alpha, \theta, \phi', \{w^\pm_i\}_i)$ with $w\in C_1$ and $w_i^\pm(\xi'_i)=u_i^\pm(\xi'_i)$ for each $i$. Indeed then,
\begin{align*}
    &\Lambda(Q',Q)=J(Q;D)- J(Q';D)\\
    &=\kappa(\xi_A^*, \xi_B^*)d(Q', Q|\xi_A^*, \xi_B^*)+\kappa(\xi_A^*, \xi_B')d(Q', Q|\xi_A^*, \xi_B')\\ &+\kappa(\xi_A', \xi_B^*)d(Q', Q|\xi_A', \xi_B^*).
\end{align*}
We have the following two sub-cases.\\
\textit{Case 3a:} $u_A^+(\xi_A^*)=u_B^+(\xi'_B)$ so that
$$Q(u_A=u_B|\xi_A^*, \xi_B^*)=Q^I(u_A=u_B|\xi_A^*, \xi_B^*),$$ $$Q(u_A=u_B|\xi_A^*, \xi_B')=Q^{II}(u_A=u_B|\xi_A^*, \xi_B')$$ and
$Q(u_A=u_B|\xi_A', \xi_B^*)=Q^{II}(u_A=u_B|\xi_A', \xi_B^*)$.
Set $w_A^+(\xi^*_A)=w_B^+(\xi_B')=u_B^+(\xi_B')$ and $\phi'_{i(\xi_i')}\equiv \phi_{i(\xi_i)}$. Now if $\phi_{a(\xi_A^*)}=\phi_{b(\xi_B^*)}=0$, we find upon straightforward evaluation that $\Lambda(Q',Q)=-\sin^2\frac{\alpha}{2}(\kappa(\xi_A^*, \xi_B')\cos\phi_{b(\xi_B')}
+\kappa(\xi_A', \xi_B^*)\cos\phi_{a(\xi_A')})$
and that if $\phi_{a(\xi_A^*)}=\phi_{b(\xi_B^*)}=\pi$, $\Lambda(Q',Q)=
\cos^2\frac{\alpha}{2}(\kappa(\xi_A^*, \xi_B')\cos\phi_{b(\xi_B')}
+\kappa(\xi_A', \xi_B^*)\cos\phi_{a(\xi_A')})$
so one of the above two ensure $\Lambda\leq 0$.
\\
\textit{Case 3b:} $u_A^+(\xi_A^*)\neq u_B^+(\xi_B')$ in which case, $$Q(u_A=u_B|\xi_A^*, \xi_B^*)=Q^I(u_A=u_B|\xi_A^*, \xi_B^*),$$ $$Q(u_A=u_B|\xi_A^*, \xi_B')=Q^{III}(u_A=u_B|\xi_A^*, \xi_B')$$ and
$$Q(u_A=u_B|\xi_A', \xi_B^*)=Q^{III}(u_A=u_B|\xi_A', \xi_B^*).$$

Again set $w_A^+(\xi^*_A)=w_B^+(\xi_B')=u_B^+(\xi_B')$ and $\phi'_{i(\xi_i')}\equiv \phi_{i(\xi_i)}$. Now if $\phi_{a(\xi_A^*)}=0,\phi_{b(\xi_B^*)}=\pi$, we find upon straightforward evaluation that $\Lambda(Q',Q)=
-\cos^2\frac{\alpha}{2}(\kappa(\xi_A^*, \xi_B')\cos\phi_{b(\xi_B')}
+\kappa(\xi_A', \xi_B^*)\cos\phi_{a(\xi_A')})$
and that if $\phi_{a(\xi_A^*)}=\pi, \phi_{b(\xi_B^*)}=0$, $\Lambda(Q',Q)=$,
$\sin^2\frac{\alpha}{2}(\kappa(\xi_A^*, \xi_B')\cos\phi_{b(\xi_B')}
+\kappa(\xi_A', \xi_B^*)\cos\phi_{a(\xi_A')})$
so one of the above two ensure $\Lambda\leq 0$.

Now note that $Q'\in C_1$ in both the sub-cases. Hence, there exists a $Q_0(\alpha^0, \theta^0, \phi^0)$  can be found, as follows from the proof of \textit{Case 1}. This completes the proof for \textit{Case 3}.

\textit{Case 4:} $\{\exists$ a unique $j\in\{A, B\}$ such that $u_j^+(\xi_j)=u_j^-(\xi_j)$ $\forall\ \xi_j$ and $u_i^+(\xi_i)\neq u_i^-(\xi_i)\ \forall \xi_i$, for $i\neq j\}$.

$|C_4|=112$. We argue for $j=A$, and $j=B$ similarly follows. We have
$u_A^+(\xi_A)=u_A^-(\xi_A)=u(\xi_A)\ \forall\xi_A$ then notice from Table \ref{tab:alphaprobs} from the main paper (Proposition~\ref{prop:alphaprob}),
$Q(u_A, u^{+}_B(\xi_B)|\xi_A, \xi_B)=\left(\cos^2\frac{\alpha}{2}\cos^2\frac{\phi_{b(\xi_B)}}{2}+\sin^2\frac{\alpha}{2}\sin^2\frac{\phi_{b(\xi_B)}}{2}\right)\delta_{u_A u(\xi_A)}$\\
$Q(u_A, u^{-}_B(\xi_B)|\xi_A, \xi_B) =(\cos^2\frac{\alpha}{2}\sin^2\frac{\phi_{b(\xi_B)}}{2}$\\
$+\sin^2\frac{\alpha}{2}\cos^2\frac{\phi_{b(\xi_B)}}{2})\delta_{u_A u(\xi_A)}$

Now consider a $Q'=(0, \theta', \phi', \{w_i^\pm(\xi_i)\}_i)$ where $\theta'_{i(\xi_i)}=\theta_{i(\xi_i)}\ \forall\ i$, $\phi_{a(\xi_A)}\equiv 0$, $w\in C_1$, $w^+_A(\xi_A)=u(\xi_A)\neq w^-_A(\xi_A)$, $w_B^\pm(\xi_B)\equiv u^\pm_B(\xi_B)$ and $\phi'_{b(\xi_B)}$ is set by the equation
$$\cos^2\phi'_{b(\xi_B)}=\left(\cos^2\frac{\alpha}{2}\cos^2\frac{\phi_{b(\xi_B)}}{2}+\sin^2\frac{\alpha}{2}\sin^2\frac{\phi_{b(\xi_B)}}{2}\right).$$
The existence of such a $\phi'$ is guaranteed since RHS in above equation is in $[0, 1]$. It is then immediate by inspection (again using Table \ref{tab:alphaprobs}) that $Q'\equiv Q\implies J(Q';D)=J(Q;D)\implies \Lambda(Q',Q)\leq 0$. Again, since $Q'\in C_1$, the required $Q_0(\alpha^0, \theta^0, \phi^0)$ can be found to settle case 4. This completes our proof.

\section{Proof of Theorem \ref{thm:tightboundsCAC}}\label{proof:tightboundsCAC}
Define,
$\Delta^*(\chi):=\inf_{Q\in\Qscr}J(Q;D)-\inf_{\pi\in\Pi}J(\pi;D)\}.$
Then from \eqref{eq:detoptimas}  and \eqref{eq:phioptimum} (from the main paper), we have $\Delta^*(\chi)\leq \Delta(\chi, \phi)$ for all $\chi>0, \phi$ where we define $\Delta(\chi, \phi):=$
\begin{multline*}
    \begin{cases}
        \underline{J}(\chi, \phi_{a1}, \phi_{b0}, \phi_{b1})+(s+k+2t)&\chi\in \left(0, \frac{k+t}{s+t}\right]\\
        \overline{J}(\chi, \phi_{a1}, \phi_{b0}, \phi_{b1})+(1+\chi)(s+t)&\chi\in \left(\frac{k+t}{s+t},\frac{s+t}{k+t}\right)\\
        \overline{J}(\chi, \phi_{a1}, \phi_{b0}, \phi_{b1})+\chi(s+k+2t)&\chi\in \left[\frac{s+t}{k+t}, \infty\right)
    \end{cases}
\end{multline*}
where $\overline J(\chi, \phi_{a1}, \phi_{b0}, \phi_{b1})$ and $\underline J(\chi, \phi_{a1}, \phi_{b0}, \phi_{b1})$ are given by \eqref{eq:phipar} and \eqref{eq:phipar2} (from the main paper) respectively.
The proof idea is as follows.
\begin{itemize}
    \item We show existence of a $\phi$ such that $\Delta(\chi, \phi)<0$ which implies $\Delta^*(\chi)<0$ corresponding to the presence of quantum advantage at $\chi$, for all $\chi\in\Xscr$.

    \item  We show $\Delta(\chi, \phi)\geq 0\forall\ \phi\in [0, \pi], \chi\in\Xscr^c$ which asserts the absence of quantum advantage for all $\chi\in \Xscr^c$ (Recall from Corollary \ref{prop:pienough} and \eqref{eq:phioptimum} (from the main paper) that $\Delta^*(\chi)$ is attained by a tuple $\phi\in[0, \pi]$).
\end{itemize}
1) Consider the interval $\chi\in\left(\frac{k+t}{s+t},\frac{s+t}{k+t}\right)$.  $\Delta^*(\chi)\leq \Delta(\chi,  \phi_{a1}, \phi_{b0}, \phi_{b1}):=\overline{J}(\chi, \phi_{a1}, \phi_{b0}, \phi_{b1})+(1+\chi)(s+t)=$
\begin{align*}
    &\left(\frac{-\chi k+s}{2}\right)\left(1-\cos{\phi_{b0}}\right) +\frac{t(\chi -1)}{2}\left(\cos{\phi_{b1}}+\cos{(\phi_{a1}+\phi_{b0})}\right)\\
    &+\left(\frac{\chi s-k}{2}\right)\left(1+\cos{(\phi_{a1}+\phi_{b1})}\right)
\end{align*}
Notice $\Delta(\chi, \pi, 0, 0)=\nabla_\phi \Delta(\chi, \pi, 0, 0)=0\ \forall\chi$ and consider the Hessian evaluated at this point is
$\begin{pmatrix}
    \frac{-\chi k + s+t (\chi  - 1)}{2} & \frac{t (\chi  - 1)}{2} & 0\\
    \frac{t (\chi  - 1)}{2} & \frac{t (\chi  - 1)+(\chi s - k)}{2} & \frac{(\chi s - k)}{2}\\
    0 & \frac{(\chi s - k)}{2} &  \frac{(\chi s - k)+t (\chi - 1)}{2}
\end{pmatrix}$.\\
The determinant of the Hessian evaluates to $\frac{1}{8}(-1 + \chi )^2 (1 + \chi ) (k - s) t^2$
which is negative for all $\chi$ except $\chi=1$ since $s>k$.
Thus, the point $(\chi, \pi, 0, 0)$ is not a local minimum and by Taylor's theorem, $\exists (\chi, \phi^*)$ such that  $\Delta(\chi, \phi^*)<0$ so that $\forall\chi\in\left(\frac{k+t}{s+t}, \frac{s+t}{k+t}\right)\setminus \{1\}\implies \Delta^*(\chi)<0$.
To show that $\Delta^*(1)\geq 0$, notice that $\Delta(1, \phi_{a1}, \phi_{b0}, \phi_{b1})=((s-k)/2)(2-\cos\phi_{b0}+\cos(\phi_{a1}+\phi_{b1}))\geq 0 \forall \phi$. We have shown that $D$ admits a quantum advantage for all $\chi \in \left(k+t/s+t, s+t/k+t\right)\setminus\{1\}=:\Xscr_1$.

2) Now for the interval $\chi\in\left[\frac{s+t}{k+t}, \infty\right)$ where $\Delta^*(\chi) \leq \Delta(\chi,  \phi_{a1}, \phi_{b0}, \phi_{b1})=\overline{J}(\chi, \phi_{a1}, \phi_{b0}, \phi_{b1})+\chi(k+s+2t):=$
\begin{align}
    &\left(\frac{\chi k-s}{2}\right)\left(1+\cos{\phi_{b0}}\right) +\left(\frac{\chi s-k}{2}\right)\left(1+\cos{(\phi_{a1}+\phi_{b1})}\right)\nonumber\\
    &+\frac{t(\chi -1)}{2}\left(2+\cos{\phi_{b1}}+\cos{(\phi_{a1}+\phi_{b0})}\right)\label{eq:delta2}
\end{align}
Notice that $\Delta(\chi, 0, \pi, \pi)=\nabla_\phi \Delta(\chi, 0, \pi, \pi)=0\forall\chi$. Further the Hessian evaluated at this point is
$\begin{pmatrix}
    \frac{\chi k - s+t (\chi  - 1)}{2} & \frac{t (\chi  - 1)}{2} & 0\\
    \frac{t (\chi  - 1)}{2} & \frac{t (\chi  - 1)+(\chi s - k)}{2} & \frac{(\chi s - k)}{2}\\
    0 & \frac{(\chi s - k)}{2} & \frac{t (\chi  - 1)+(\chi s - k)}{2}
\end{pmatrix}$
and its determinant given by	$(\frac{1}{8}) (\chi -1) t (2 (\chi  k - s) (\chi s-k) + (-1 + \chi )^2 (k + s) t).$
Since $s>k$, $\chi>1$, the Hessian has a negative determinant for all $\chi$ satisfying
\begin{equation}\label{eq:quaddef}
    \overline f(\chi):=(2 (\chi  k - s) (\chi s-k) + (-1 + \chi )^2 (k + s) t)<0.
\end{equation}
This translates to $\chi<\chi^{th}$ since $\chi^{th}$ is the larger root of the quadratic $\overline{f}(\chi)$.
and we thus have existence of a $ (\chi, \phi^*)$ such that $\Delta(\chi, \phi^*)<0$ so that $\forall\chi\in\left[\frac{s+t}{k+t}, \chi^{th}\right),  \Delta^*(\chi)<0$.
Next, we show that $\Delta^*(\chi)\geq 0$ $\forall \chi>\chi^{th}$.
In Proposition \ref{prop:globalmin} below, we have shown that the global minimum of $ \Delta(\chi^{th}, \phi)$ in the cube $\phi \in [0, \pi]^3$  occurs at one of the vertices. Substituting each of the vertices in \eqref{eq:delta2} and using the fact that $\chi^{th}>\frac{s+t}{k+t}>1>\frac{k+t}{s+t}$
reveals that the global minimum is attained at $\phi=(0, \pi, \pi)$ \ie $	\Delta(\chi^{th}, 0,\pi,\pi)= 0$. We skip the explicit evaluations due to paucity of space.

This establishes $\Delta^*(\chi)=0$ for all $\chi\geq \chi^{th}$ since for any $\phi$, $\partial_{\chi} \Delta(\chi, \phi)=
\frac{1}{2}(k(1+\cos\phi_{b0})+t(2+\cos\phi_{b1}+\cos(\phi_{a1}+\phi_{b0}))\\
s(1+\cos(\phi_{a1}+\phi_{b1})))\geq 0\forall\phi$.  Our arguments have thus shown that quantum advantage is present for $\chi\in [s+t/k+t, \chi^{th})=:\Xscr_2$ and absent for $\chi>\chi^{th}$.

3) Finally for the interval $\left(0, \frac{k+t}{s+t}\right)$, $\Delta^*(\chi)\leq \Delta(\chi,  \phi_{a1}, \phi_{b0}, \phi_{b1})=$
\begin{align}
    & (\frac{\chi k-s}{2})(1+\cos\phi_{b0})+(\frac{\chi t-t}{2})(2+\cos\phi_{b1}+\cos(\phi_{a_1}-\phi_{b_0})) \nonumber\\
    &+(\frac{\chi s-k}{2})(1+\cos(\phi_{a1}+\phi_{b1})) -(\chi-1)(s+k+2t)\label{eq:delta3}
\end{align}
Again, $\Delta(\chi, 0,0,0)=\nabla_\phi \Delta(\chi, 0,0,0)=0$. Further, Hessian evaluated at this point is
$\begin{pmatrix}
    \frac{\chi k - s}{2} + \frac{t (\chi  - 1)}{2} & \frac{t (\chi  - 1)}{2} & 0\\
    \frac{t (\chi  - 1)}{2} & \frac{t (\chi  - 1)}{2} + \frac{(\chi s - k)}{2} & \frac{(\chi s - k)}{2}\\
    0 & \frac{(\chi s - k)}{2} & \frac{t (\chi  - 1)}{2} + \frac{(\chi s - k)}{2}
\end{pmatrix}$
and has the determinant$
(\frac{1}{8} )(1-\chi) t (2 (\chi  k - s) (\chi s-k) + (-1 + \chi )^2 (k + s) t).$
Thus, the Hessian has a negative eigenvalue for
$\overline f(\chi)<0.$
This translates to $\chi>\chi_{th}$ since $\chi_{th}$ is the smaller root of $\overline f(\chi)$. Rest follows as before and we have $\Delta^* (\chi)<0$ for $\chi\in (\chi_{th}, 1)$.
For $\chi<\chi_{th}$, the point $(\chi, 0, 0, 0)$ turns into a global minimum, and quantum advantage is lost for $\chi\leq \chi_{th}$. Our arguments that show this mirror those in part 2). We have shown in Proposition \ref{prop:globalmin} below that the global minimum of $\Delta(\chi_{th}, \phi)$ occurs on one of the vertices of the cube $\phi\in [0, \pi]^3$. We again find using $\chi_{th}<\frac{k+t}{s+t}<1<\frac{s+t}{k+t}$ and substituting each vertex in \eqref{eq:delta3} that the global minimum occurs at the $\phi=0,0,0$ and $\Delta^*(\chi_{th})=0$.

Now we can conclude $\Delta^*(\chi)=0$ for all $\chi\leq \chi_{th}$ since $\partial_{\chi} \Delta(\chi, \phi)=
\frac{1}{2}(k (-1 + \cos\phi_{b0}) +
t (-2 +\cos\phi_{a1} - \cos\phi_{b0} + \cos\phi_{b1})
+ s (-1 + \cos\phi_{a1} +\phi_{b1}))\leq 0\ \forall\phi$.
This demonstrates that the quantum advantage is present for $\chi\in (\chi_{th}, (k+t)/(s+t)]$ and absent for all $\chi<\chi_{th}$.

\begin{proposition}
    1) Suppose $\Delta(\chi^{th}, \phi)$ as defined in \eqref{eq:delta2} attains a global minimum at $\phi^*\in [0, \pi]^3$. Then $\phi^*$ is a vertex of the cube $[0, \pi]^3$.\\
    2) Suppose $\Delta(\chi_{th}, \phi)$ as defined in \eqref{eq:delta3} attains a global  minimum at $\phi^*\in [0, \pi]^3$. Then $\phi^*$ is a vertex of the cube $[0, \pi]^3$. \label{prop:globalmin}
\end{proposition}
\begin{proof}
    1) It is easy to check $\overline f(s/k)>0$ which implies  $s/k>\chi^{th}$ so let $\alpha^2=-{(\chi^{th} k-s)}/{2}$, $\beta^2=t(\chi^{th} -1)/2$ and $\delta^2=(\chi^{th} s-k)/2$. We then have (for brevity we change notation to $\Delta(\phi):=\Delta(\chi^{th}, \phi)$) the system $\nabla_\phi\Delta(\phi)$ given by,
    \begin{align}\label{eq:sys1}
        &\partial_{\phi_{b0}} \Delta=\alpha^2 \sin{\phi_{b0}}-\beta^2\sin(\phi_{a1}+\phi_{b0})=0\\
        \label{eq:sys2}
        &\partial_{\phi_{a1}} \Delta=-\beta^2\sin(\phi_{a1}+\phi_{b0})-\delta^2\sin(\phi_{a1}+\phi_{b1})=0\\
        \label{eq:sys3}
        &\partial_{\phi_{b1}} \Delta=-\beta^2\sin(\phi_{b1})-\delta^2\sin(\phi_{a1}+\phi_{b1})=0
    \end{align}
    We now find all solutions on $C:=[0, \pi]^3$. It is easy to notice that each vertex of $C$ is  a solution. Further we also notice that the vertices are the only solution on each of the faces of $C$. To check this, notice how $\phi_{a1}\in \{0, \pi\}$ forces $\phi_{b0}, \phi_{b1}\in\{0, \pi\}$ from \eqref{eq:sys1} and \eqref{eq:sys3} and repeat this argument for other faces.

    Now, from $\eqref{eq:sys1}$, notice that $\sin\phi_{b0}\geq 0\implies \sin(\phi_{a1}+\phi_{b0})\geq 0\implies \phi_{a1}+\phi_{b0}\in [0,\pi]$. Comparing \eqref{eq:sys2} and \eqref{eq:sys3}, we obtain $\sin\phi_{b1}=\sin(\phi_{a1}+\phi_{b0})$ so that $\phi_{b1}\in \{\phi_{a1}+\phi_{b0}, \pi-(\phi_{a1}+\phi_{b0})\}$. Suppose that $\phi_{b1}= \pi-(\phi_{a1}+\phi_{b0})$ and substitute so in \eqref{eq:sys2} to obtain $\sin\phi_{b0}/\sin(\phi_{a1}+\phi_{b0})<0$ which contradicts \eqref{eq:sys1}. Thus $\phi_{b1}=\phi_{a1}+\phi_{b0}$, so substitute in \eqref{eq:sys2} and using $b^2=\beta^2/\alpha^2$ and $d^2=\delta^2/\beta^2$ , rewrite \eqref{eq:sys1} and \eqref{eq:sys2} respectively as $\sin\phi_{b0}=b^2\sin(\phi_{a1}+\phi_{b0}),$
    $\sin(\phi_{a1}+\phi_{b0})+d^2\sin(2\phi_{a1}+\phi_{b0})=0.$ from where it follows that
    \begin{equation}\label{eq:sys4}
        \sin\phi_{b0}=-b^2d^2\sin{(\phi_{b0}+\phi_{a1})}.
    \end{equation}
    We now show there is no simultaneous solution to \eqref{eq:sys1}, \eqref{eq:sys2} and \eqref{eq:sys3} in $(0, \pi)^3$. Suppose there is one. Expanding sinusoidal sums in \eqref{eq:sys4} and division by $\sin\phi_{b0}$ recovers $b^2(\cos\phi_{a1}+\cot{\phi_{b0}}\sin{\phi_{a1}})=1$. This allows the following simplification.
    \begin{align*}
        &-b^2d^2(\cos2\phi_{a1}+\cot{\phi_{b0}}\sin{2\phi_{a1}})=1\\
        & -b^2d^2(-1+2\cos{\phi_{a1}}(\cos\phi_{a1}+\cot\phi_{b0}\sin\phi_{a1}))=1\\
        & \cos\phi_{a1}=\frac{1}{2}(b^2-1/d^2)=\frac{1}{2}\left(\frac{\beta^2}{\alpha^2}-\frac{\beta^2}{\delta^2}\right)
    \end{align*}
    Now recall $\overline f(\chi)$ from \eqref{eq:quaddef} and observe $\overline{f}(\chi^{th})=0$ that asserts the following train of calculations, using substitutions for $\alpha^2, \beta^2$ and $\delta^2$
    \begin{align*}
        &(2 (\chi  k - s) (\chi s-k) + (-1 + \chi )^2 (k + s) t)=0\\
        & \frac{t(\chi-1)^2(s+k)}{2(s-\chi k)(\chi s-k)}=1\\
        & \frac{1}{2}(b^2-1/d^2)=\cos\phi_{a1}=1\implies \phi_{a1}=0.
    \end{align*}
    This contradicts $\phi\in (0,\pi)^3$. We have thus shown that the solution set to $\nabla_{\phi}\Delta(\phi)=0$ is precisely the set of vertices of $C$. Extreme value theorem thus dictates that a global minimum of $\Delta(\phi)$ is attained on the boundary of $C$. Now similar arguments can be repeated on a face of $C$ to push the location of the global minimizer to the edges, and subsequently the vertices. We show one such example and the rest follows similarly. So look at the face $\phi_{a1}=0$, let $\Delta(0, \phi_{b0}, \phi_{b1})=\delta(\phi)$. Then substituting $\phi_{a1}$ in \eqref{eq:sys1} and \eqref{eq:sys3} forces the solution to $\nabla_\phi \delta=0$ on a vertex. Thus, $\delta$ attains a minimum on one of the edges of the face $\phi_{a1}=0$. So to look at the edge  $\phi_{a1}=0, \phi_{b1}=0$, substitute this in \eqref{eq:sys1} and notice that $\Delta(\chi^{th}, 0, \phi_{b0}, 0)$ attains a minimum on a vertex of C. It is straightforward to repeat this argument for all faces and edges which will establish the proposition.

    2) The proof scheme is exactly the same. We prove the less trivial part again, i.e. to show that $\nabla_\phi \Delta(\chi_{th}, \phi)=0$ has no solution in the interior of $C$. Note that $\overline{f}(k/s)<0$ so $1>\frac{k+t}{s+t}>\chi_{th}>k/s$ and let $\alpha^2=-{(\chi_{th} k-s)}/{2}$, $\beta^2=-t(\chi_{th} -1)/2$ and $\delta^2=(\chi_{th} s-k)/2$.Then  (with the notation $\Delta(\phi):=\Delta(\chi_{th}, \phi)$) the system $\nabla_\phi\Delta(\phi)$ given by,
    \begin{equation}\label{eq:syss1}
        \partial_{\phi_{b0}} \Delta=\alpha^2 \sin{\phi_{b0}}-\beta^2\sin(\phi_{a1}-\phi_{b0})=0
    \end{equation}
    \begin{equation}\label{eq:syss2}
        \partial_{\phi_{a1}} \Delta=\beta^2\sin(\phi_{a1}-\phi_{b0})-\delta^2\sin(\phi_{a1}+\phi_{b1})=0
    \end{equation}
    \begin{equation}\label{eq:syss3}
        \partial_{\phi_{b1}} \Delta=\beta^2\sin(\phi_{b1})-\delta^2\sin(\phi_{a1}+\phi_{b1})=0
    \end{equation}
    As before, all vertices of $C$ are solutions to the system and vertices are the complete set of solutions on the boundary. As before, we argue $\phi_{b1}=\phi_{a1}-\phi_{b0}$ and arrive at $\sin\phi_{b0}=b^2 \sin{\phi_{a1}-\phi_{b0}}$ and $\sin\phi_{b0}=b^2d^2 \sin(2\phi_{a1}-\phi_{b0})$. Following manipulations similar to before, one can show that this forces $\phi$ to vertex. Maintaining rest of the proof scheme completes the proof of part 2.
\end{proof}


\bibliographystyle{IEEEtran}
\bibliography{ref}

\begin{biography}[{\includegraphics[width=1in,height=1in]{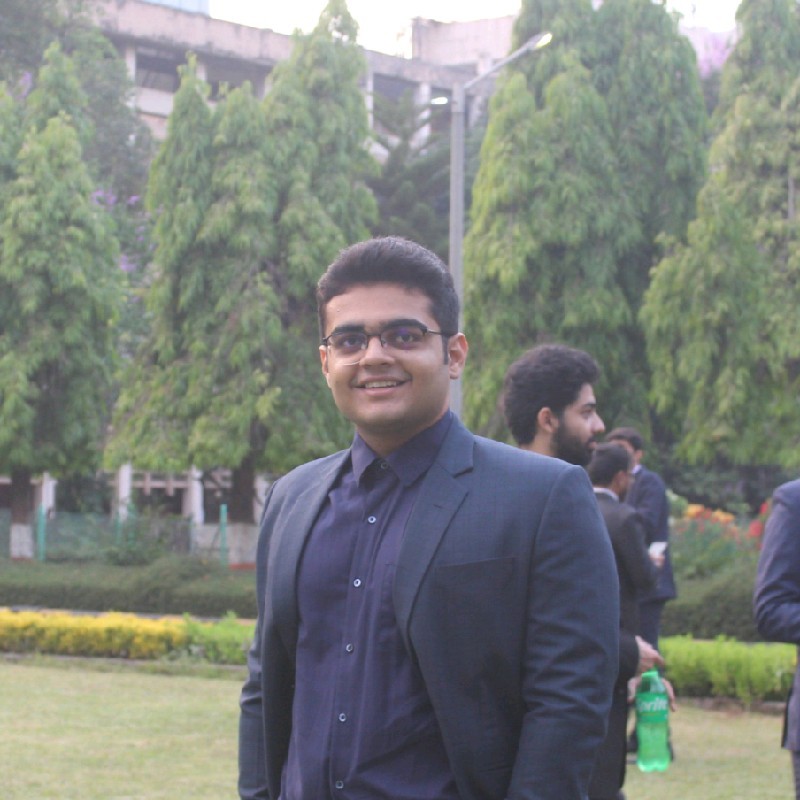}}]{Shashank Aniruddha Deshpande} a graduate student at the Massachusetts Institute of Technology in the Department of Aeronautics and Astronautics. Before this, he was an undegraduate the Department of Physics, and, the Department of Systems and Control Engineering at IIT Bombay, India. His research interests broadly span control and optimization of stochastic and networked systems.
\end{biography}
\begin{biography}[{\includegraphics[width=1in, keepaspectratio]{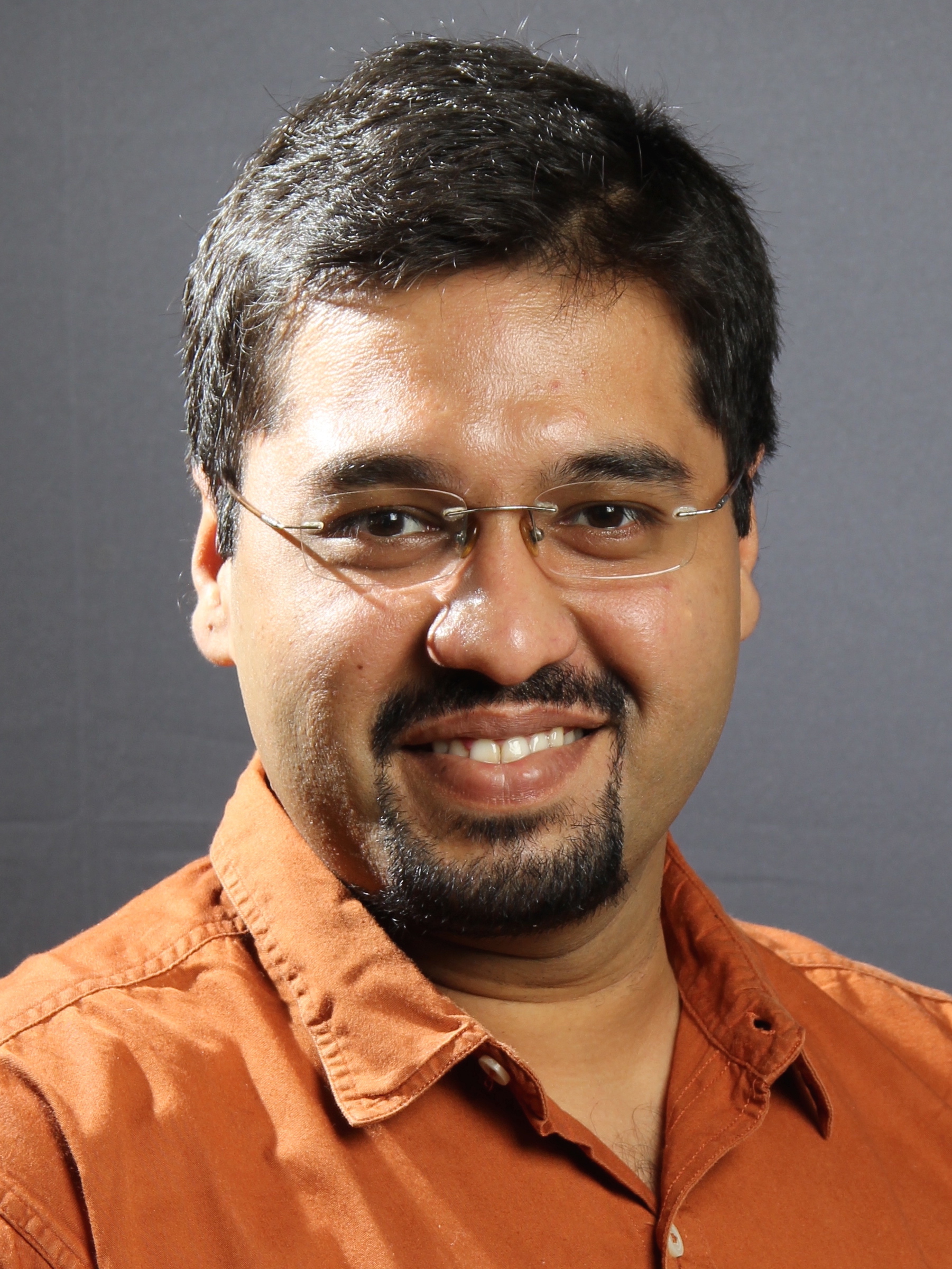}}]{Ankur A. Kulkarni}
	is the Kelkar Family Chair Professor with the Systems and Control Engineering group at the Indian Institute of Technology Bombay (IITB). He received his B.Tech. from IITB in 2006, followed by M.S. in 2008 and Ph.D. in 2010, both from the University of Illinois at Urbana-Champaign (UIUC). From 2010-2012 he was a post-doctoral researcher at the Coordinated Science Laboratory at UIUC. He was an Associate (from 2015--2018) of the Indian Academy of Sciences, a recipient of the INSPIRE Faculty Award of the Dept of Science and Technology Govt of India in 2013. He has won
	several best paper awards at conferences, the Excellence in Teaching Award in 2018 at IITB, and the William A. Chittenden Award in 2008 at UIUC.
	He has been a consultant to  the Securities and Exchange Board of India (SEBI), the HDFC Life Insurance Company, Kotak Mahindra Bank Ltd and Bank of Baroda. He presently serves on the IT-Project Advisory Board of SEBI, as Research Advisor to the Tata Consultancy Services, and as Program Chair of the Indian Control Conference. He has been a visitor to MIT in USA, University of Cambridge in UK, NUS in Singapore, IISc in Bangalore and KTH in  Sweden.
\end{biography}

\end{document}